\documentclass[11pt,a4paper]{article}
\usepackage[a4paper]{geometry}
\usepackage{amssymb,latexsym,amsmath,amsfonts,amsthm}
\usepackage{graphicx}
\usepackage{epstopdf}
\usepackage{tikz}
\usepackage{pgflibraryshapes}
\usetikzlibrary{arrows,decorations.markings}
\usepackage{subfigure}
\usepackage{overpic}
\usepackage{fullpage}
\usepackage{comment,verbatim}
\usepackage{hyperref}
\usepackage[title,titletoc]{appendix}
\usepackage{float}
\usepackage{appendix}
\usepackage{caption}
\usepackage{subfigure}
\usepackage{authblk}

\newtheorem{thm}{Theorem}[section]
\newtheorem{lemma}[thm]{Lemma}
\newtheorem{pro}[thm]{Proposition}

\newtheorem{corollary}[thm]{Corollary}

\newtheorem{rem}[thm]{Remark}

\newcommand{\dd}{\mathrm{d}}

\numberwithin{equation}{section}

\def\arg {\mathop{\rm arg}\nolimits}
\def\Re {\mathop{\rm Re}\nolimits}

\def\sup {\mathop{\rm sup}\nolimits}
\def\Res {\mathop{\rm Res}\nolimits}
\newcommand{\msf}{\mathsf}

\begin{document}

\title{Asymptotics of the deformed higher order Airy-kernel determinants and applications}


\author[1,2]{Jun Xia}
\author[3]{Yi-Fan Hao}
\author[4]{Shuai-Xia Xu}
\author[3,5]{Lun Zhang}
\author[1]{Yu-Qiu Zhao}
\affil[1]{Department of Mathematics, Sun Yat-sen University, Guangzhou 510275, China.}
\affil[2]{Department of Mathematics, Jiaying University, Meizhou 514015, China.}
\affil[3]{ School of Mathematical Sciences, Fudan University, Shanghai 200433, China.}
\affil[4]{Institut Franco-Chinois de l'Energie Nucl\'{e}aire, Sun Yat-sen University,
Guangzhou 510275, China.}
\affil[5] {Shanghai Key Laboratory for Contemporary Applied Mathematics, Fudan University, Shanghai 200433, China.}

\renewcommand{\thefootnote}{\fnsymbol{footnote}}
\footnotetext{Email addresses: junxiacqu@foxmail.com (J. Xia); yifanhao98@gmail.com (Y.-F. Hao);  xushx3@mail.sysu.edu.cn (S.-X. Xu);
lunzhang@fudan.edu.cn (L. Zhang); stszyq@mail.sysu.edu.cn (Y.-Q. Zhao)}

\date{}
\maketitle
\begin{abstract}
We study the one-parameter family of Fredholm determinants $\det(I-\rho^2\mathcal{K}_{n,x})$, $\rho\in\mathbb{R}$, where $\mathcal{K}_{n,x}$ stands for the integral operator acting on $L^2(x,+\infty)$ with the higher order Airy kernel. This family of determinants represents a new universal class of distributions which is a higher order analogue of the classical Tracy-Widom distribution. Each of the determinants admits an integral representation in terms of a special real solution to the $n$-th member of the Painlev\'{e} II hierarchy. Using the Riemann-Hilbert approach, we establish asymptotics of the determinants and the associated higher order Painlev\'{e} II transcendents as $x\to -\infty$ for $0<|\rho|<1$ and $|\rho|>1$, respectively. In the case of $0<|\rho|<1$, we are able to calculate the constant term in the asymptotic expansion of the determinants, while for $|\rho|>1$, the relevant asymptotics exhibit singular behaviors. Applications of our results are also discussed, which particularly include asymptotic statistical properties of the counting function for  the random point process defined by the higher order Airy kernel.


\end{abstract}


\noindent
\section{Introduction}
In this paper, we are concerned with a family of integral kernels defined through the following double contour integral:
\begin{equation}\label{kernal}
K_n(x,y)=\frac{1}{(2\pi i)^2}\int_{\gamma_R}\int_{\gamma_L}
\frac{\mathrm{e}^{(-1)^{n+1}(P_{2n+1}(u)-P_{2n+1}(v))-xu+yv}}{u-v}
\mathrm{d}u\mathrm{d}v,~~x,y\in\mathbb{R},
\end{equation}
where
\begin{equation}\label{eq:Pn}
P_{2n+1}(z)=\frac{z^{2n+1}}{2n+1}+\sum_{k=1}^{n-1}
\frac{\tau_{k}}{2k+1}z^{2k+1},\quad n\in \mathbb{N}, \quad \tau_1,\ldots,\tau_{n-1}\in\mathbb{R},
\end{equation}
is an odd polynomial of degree $2n+1$, the contours $\gamma_R$ and $\gamma_L$ are symmetric with respect to the imaginary axis and are asymptotic to the straight lines with arguments $\pm \frac{n}{2n+1}\pi$ at infinity; see Figure \ref{gamma-RL} for an illustration. We call $K_{n}$ the higher order Airy kernels as it reduces to the prominent Airy kernel if $n=1$. Moreover, if $P_{2n+1}$ in \eqref{kernal} is a monomial $z^{2n+1}/(2n+1)$, $K_n$ is equal to the generalized Airy kernel \cite{DMS2018} defined by
\begin{equation}\label{def:gAirykernal}
K_{\mathrm{Ai}_{2n+1}}(x,y)=\int_{0}^{\infty}\mathrm{Ai}_{2n+1}(x+s)
\mathrm{Ai}_{2n+1}(y+s)\,\mathrm{d}s,
\end{equation}
where
\begin{equation}\label{Airy}
\mathrm{Ai}_{2n+1}(x)=\frac{1}{\pi}\int_{0}^{+\infty}
\cos\left(\frac{s^{2n+1}}{2n+1}+xs\right)\mathrm{d}s
\end{equation}
is the higher order Airy function and satisfies the ordinary differential equation
$$
\frac{\dd ^{2n}}{\dd x^{2n}}y(x)=(-1)^{n+1}xy(x).
$$

\begin{figure}[H]
  \centering
  \includegraphics[width=8cm,height=8cm]{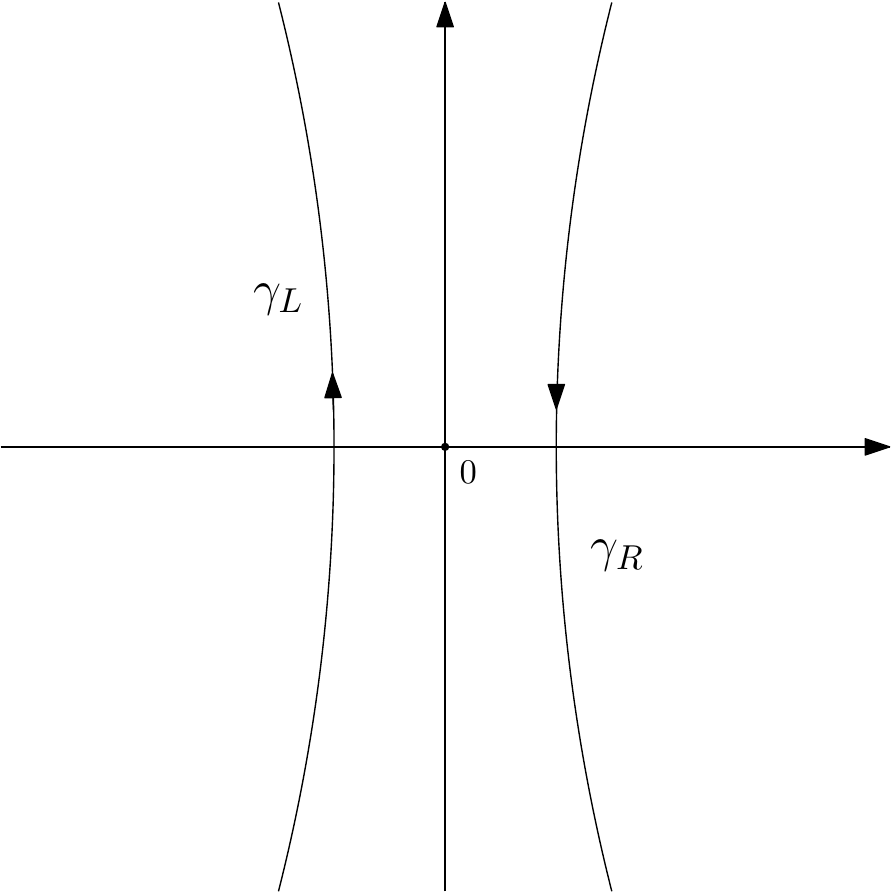}\\
  \caption{The contours $\gamma_R$ and $\gamma_L$ in the definition \eqref{kernal} of $K_n$. }\label{gamma-RL}
\end{figure}

Determinantal point processes associated with the kernels $K_n$ are believed to be universal objects in describing local statistics for a variety of interacting particle systems arising from random matrix theory and statistical physics. The Airy process characterizes the limiting distribution of the normalized largest eigenvalues for a large class of random matrices \cite{Mehta}, while the generalized Airy process corresponding to the kernel $K_{\mathrm{Ai}_{2n+1}}$  describes the multicritical edge statistics for the momenta of fermions trapped in a non-harmonic potential \cite{DMS2018}. We also refer to \cite{KZ1, KZ} for the connections between point processes of higher order Airy kernels and random partitions obeying the Schur measure.


Let $\mathcal{K}_{n,x}$ be the integral operator acting on $L^2(x,+\infty)$ with the higher order Airy kernel $K_n$. A central object of the present work is the Fredholm determinants
\begin{align}\label{Fred-Deter}
F_n(x;\rho):=\det(I-\rho^2 \mathcal{K}_{n,x}),\qquad \rho\in \mathbb{R}.
\end{align}
From the theory of determinantal point processes \cite{Joh06,Sos}, it follows that $F_n(x;\pm 1)$ is the probability distribution of the largest particle
in the process. If $|\rho|<1$, the deformed determinant $F_n(x;\rho)$ can be interpreted as the probability distribution of the largest particle in the associated thinned process. The thinned process is obtained from the original one by removing each particle independently with probability $1-\rho^2$, and is a classical operation in the studies of point processes; cf. \cite{BP06,Illbook}. It has been shown in \cite{CCG2021} that $F_n(x;\rho)$ indeed defines a distribution function for $|\rho|\leq 1$.

A celebrated result of Tracy and Widom established in \cite{TW94} shows that the Airy-kernel determinant $F_1(x; \pm 1)$ admits an integral representation in terms of the Hastings-McLeod solution \cite{HM80} of the Painlev\'{e} II equation
\begin{align}\label{eq:PII}
q''(x)=2q^3(x)+xq(x).
\end{align}
The same integral formula holds for the deformed Airy-kernel determinant $F_1(x; \rho)$, $|\rho|<1$, but involves the Ablowitz-Segur solution \cite{AS1976,AS81} of \eqref{eq:PII} therein; cf. \cite{BCP09,BB2018}. Analogous results have been obtained in \cite{CCG2021} for higher order Airy-kernel determinants $F_{n}$, where one instead encounters the Painlev\'{e} II hierarchy. The Painlev\'{e} II hierarchy is obtained from
the mKdV hierarchy via a self-similar reduction \cite{FN}; see also \cite{CJP,Kud02,Maz07}. It is defined by a sequence of ordinary differential equations and the $n$-th member is given by
\begin{equation}\label{PIIhierarchy}
\left(\frac{\mathrm{d}}{\mathrm{d} x}+2 q\right) \mathcal{L}_{n}\left[q'-q^{2}\right]+\sum_{i=1}^{n-1} \tau_{i}\left(\frac{\mathrm{d}}{\mathrm{d} x}+2 q\right) \mathcal{L}_{i}\left[q'-q^{2}\right]=x q, \qquad n\in\mathbb{N},
\end{equation}
where the prime denotes the derivative with respect to $x$, the parameters $\tau_1,\ldots,\tau_{n-1}$ are real, and $\mathcal{L}_k$, $k\in\mathbb{N}$, stands for the Lenard operator defined via the recursive relation
\begin{equation}\label{recursiverelation}
\frac{\mathrm{d}}{\mathrm{d} x} \mathcal{L}_{k+1} h=\left(\frac{\mathrm{d}^{3}}{\mathrm{~d} x^{3}}+4 h \frac{\mathrm{d}}{\mathrm{d} x}+2 h'\right) \mathcal{L}_{k} h,\quad \mathcal{L}_{0} h=\frac{1}{2}, \quad \mathcal{L}_{k} 0=0.
\end{equation}
The Painlev\'{e} II hierarchy \eqref{PIIhierarchy} includes the Painlev\'{e} II equation \eqref{eq:PII} as the first member, and the second member is
\begin{equation}
q''''(x)=10q(x)(q'(x))^2+10q^2(x)q''(x)-6q^5(x)-\tau_1(q''(x)-2q^3(x))+xq(x). \label{PII2}
\end{equation}
By \cite[Theorem 1.1]{CCG2021}, it comes out that the determinant $F_n$ is directly related to the Painlev\'{e} II hierarchy via the relation
\begin{equation}\label{dF1-expression}
\frac{\mathrm{d}^2}{\mathrm{d}x^2}\ln F_n(x;\rho)=-q_{n}^2\left((-1)^{n+1}x;\rho\right),
\end{equation}
where $q_{n}\left(x;\rho\right)=q_n(x;\rho;\tau_1,\ldots,\tau_{n-1})$ is a
real solution for the $n$-th member of \eqref{PIIhierarchy} 
with asymptotic behavior
\begin{equation}\label{qAsymp+}
q_{n}\left((-1)^{n+1}x;\rho\right)= O\left(
\exp\left(-C x^{\frac{2n+1}{2n}}\right) \right) \quad \mathrm{as}~~x\to+\infty,
\end{equation}
for some $C>0$.
Here the error term is uniform for $\tau_1,\dots,\tau_{n-1}$ in any compact subset of the real axis and $q_n(x;\rho)$ is  pole-free on real axis for $|\rho|\leq1$.
If $\tau_1=\cdots=\tau_{n-1}=0$, a refinement of \eqref{qAsymp+} is
\begin{equation}\label{qAymp+reduced}
q_{n}\left((-1)^{n+1}x;\rho\right)=\rho\mathrm{Ai}_{2n+1}(x)(1+o(1))\quad \mathrm{as}~~x\to+\infty,
\end{equation}
where $\mathrm{Ai}_{2n+1}$ is defined in \eqref{Airy}; see \cite[Equation (1.13)]{CCG2021}. It is worth noting that although the formulas \eqref{dF1-expression}--\eqref{qAymp+reduced} are derived in \cite{CCG2021} for $|\rho|\leq 1$, the results can be generalized to all real $\rho$ without any difficulty. If $|\rho|\leq 1$, $q_n(x;\rho)$ is a natural generalization of the Hastings-McLeod solution (for $|\rho|=1$) or the
Ablowitz-Segur solution (for $|\rho|<1$) to the Painlev\'{e} II equation \eqref{eq:PII}, which is pole free on the real axis. By integrating
\eqref{dF1-expression} twice, one further obtains the following
Tracy-Widom type formula
\begin{equation}\label{F1-expression}
F_n(x;\rho)=\exp\left\{-\int_x^{+\infty}(s-x)
q_{n}^2\left((-1)^{n+1}s;\rho\right)\mathrm{d}s\right\}.
\end{equation}
This result has been extended to several disjoint intervals in \cite{CT2021}, which is related to a vector-valued Painlev\'e II hierarchy. We also refer to \cite{BCT} for the studies of higher order  Airy process at finite temperature, where  the law is governed by a Painlev\'{e} II integro-differential hierarchy.

The main aim of this paper is to derive asymptotic behaviors of $q_{n}\left((-1)^{n+1}x;\rho\right)$ and $F_n(x;\rho)$ as $x\to -\infty$ for $|\rho| \neq 1$. If $|\rho|=1$, the relevant results have already been established in \cite{CCG2021}. More precisely, by \cite[Theorem 1.6]{CCG2021}, it follows that, as $x\to -\infty$,
\begin{equation*}\label{q1}
q_{n}\left((-1)^{n+1}x;1\right)=\sum_{i=0}^{2n}\theta_i|x|^{\frac{1-2i}{2n}}+\frac{\msf d_n}{2\theta_0}|x|^{-2-\frac{1}{2n}}
+O\left(|x|^{-2-\frac{1}{n}}\right),
\end{equation*}
and
\begin{equation*}\label{F-asymptotic-infty3}
F_n(x;\pm 1)=\msf C_n|x|^{\msf d_n}
\exp\left(-\sum_{j=0,
j\neq n+1}^{2n}\frac{n^2}{(n+1-j)(2n+1-j)}\sum_{i=0}^j\theta_i\theta_{j-i}|x|^{\frac{2n-j+1}{n}}
\right)\left(1+o(1)\right),
\end{equation*}
where
$$
\msf d_n=\left\{
      \begin{array}{ll}
        -\frac18, & \hbox{$n=1$,}\vspace{1.5mm}
 \\

        -\frac12, & \hbox{otherwise,}
      \end{array}
    \right.
$$
$\theta_0=\begin{pmatrix} 2n \\ n \end{pmatrix}^{-\frac{1}{n}}$, $\theta_i=\theta_i(\tau_1,\ldots,\tau_{n-1},n)$, $i\in \mathbb{N}$, are constants depending explicitly on $n$ and the parameters  $\tau_1,\ldots,\tau_{n-1}$ (see \cite[Equation (1.16)]{CCG2021}), and $\msf C_n>0$ is an undetermined constant. As we will see in what follows, the asymptotics of $q_{n}\left((-1)^{n+1}x;\rho\right)$ and $F_n(x;\rho)$ exhibit significantly different behaviors if $|\rho|\neq 1$.

\section{Main results}
To state our results, we set
\begin{equation}\label{ell}
Q(z)=(-1)^nP_{2n+1}'(iz)=z^{2n}+\sum_{j=1}^{n-1}(-1)^{n+j} \tau_j z^{2j}, \qquad n\in\mathbb{N},
\end{equation}
with $P_{2n+1}$ given in \eqref{eq:Pn}, and define
\begin{equation}\label{ak-expression}
a_0=1,\qquad a_{k}:=a_{k}(\tau_1,\ldots,\tau_{n-1};n)=\frac{1}{2k-1} \Res \limits_{z=\infty}Q^{\frac{2k-1}{2n}}(z), \quad k\in \mathbb{N},
\end{equation}
where the branch for $Q^{\frac{2k-1}{2n}}$ is taken such that $Q^{\frac{2k-1}{2n}}(z)\sim z^{2k-1} $ as $z\to\infty$. Note that if $\tau_1=\cdots=\tau_{n-1}=0$, we have $a_k=0$ for $k\in \mathbb{N}$. By \eqref{Fred-Deter} and \eqref{eq:PsiSymm} below, it is readily seen that
\begin{equation}\label{symmetry}
F_n(x;\rho)=F_{n}(x;-\rho),\qquad q_{n}(x;-\rho)=-q_{n}(x;\rho).
\end{equation}
This, together with the facts that $F_n(x;0)=1$ and $q_{n}(x;0)=0$ (see \eqref{dF1-expression}), implies that it suffices to focus on the case that $\rho>0$. We now state our results for $0<\rho<1$ and $\rho >1$, respectively.

\subsection{Asymptotics of $q_n((-1)^{n+1}x;\rho)$ and $F_n(x;\rho)$ for $0< \rho<1$}
As aforementioned, if $0<\rho<1$, $q_n(x;\rho)$ generalizes the Ablowitz-Segur solution of the Painlev\'{e} II equation \eqref{eq:PII}. Our first result gives large negative $x$ asymptotics of  $q_{n}\left((-1)^{n+1}x;\rho\right)$, which particularly establishes the connection formulas for this family of special solutions.

\begin{thm}\label{q-asymptotic}
Let $0<\rho<1$ and $q_{n}(x;\rho)$ be defined through \eqref{dF1-expression}, which is the solution of the $n$-th member of the Painlev\'e II hierarchy \eqref{PIIhierarchy}
subject to the boundary condition \eqref{qAsymp+}.
We have, as $x\to-\infty$,
\begin{multline}\label{qAymp-1}
q_{n}\left((-1)^{n+1}x;\rho\right)=\frac{\sqrt{\beta}}{\sqrt{n}\,|x|^{\frac{2 n-1}{4 n}}}
\cos\Bigg(\sum_{k=0}^{n}\frac{2na_k}{1+2(n-k)}|x|^{\frac{1+2(n-k)}{2n}}
-\frac{2n+1}{4n}\beta\ln|x|+\phi\Bigg)
\\
+O\left(|x|^{-\frac{2n+1}{4n}}\right),
\end{multline}
where $a_k$, $k=0,\ldots,n$, are defined in \eqref{ak-expression}, $\beta$ and $\phi$ are related to $\rho$ via the connection formulas
\begin{equation}\label{beta}
\left\{\begin{aligned}
&\beta=-\frac{1}{\pi}\ln\left(1-\rho^2\right ),\\
&\phi=-\frac{\beta}{2}\ln(8n)+\arg\Gamma\left(\frac{\beta}{2} i\right)+ \frac{\pi}{4},
\end{aligned}\right.
\end{equation}
with $\Gamma$ being the Gamma function.
\end{thm}

\begin{rem}
If $n=1$, the above theorem is due to \cite{AS1976,AS81}; see also \cite{CM-1988,DZ1995,HM80} for the studies using different methods. For general $n$ with $\tau_1=\cdots=\tau_{n-1}=0$, the results can be found in
\cite[Theorem 1.1]{HZ2021}.
\end{rem}

We next show the large deformations of the higher order Airy-kernel determinants, up to and including the constant term.
%
\begin{thm}\label{thm:largegap}
Let  $0\leq \rho<1$ and  $F_n(x;\rho)$ be the Fredholm determinant defined in \eqref{Fred-Deter}. As $x\to-\infty$,  we have
\begin{equation}\label{F-asymptotic-infty2}
F_n(x;\rho)=C_0(8n)^{\frac{\beta^2}{4}}|x|^{\frac{(2n+1)\beta^2}{8n}}\exp\left(
-\beta\sum_{k=0}^{n}\frac{2na_k}{1+2(n-k)}|x|^{\frac{1+2(n-k)}{2n}}\right)
\left[1+O\left(|x|^{-\frac{1}{2n}}\right)\right],
\end{equation}
uniformly for $\rho$ in any compact subset of $[0,1)$, where $\beta$ is given in \eqref{beta}, $a_k$, $k=0,\ldots,n$, are defined in \eqref{ak-expression}, and
\begin{equation}\label{C0}
 C_0=G\left(1+\frac{i\beta}{2}\right)G\left(1-\frac{i\beta}{2}\right),
\end{equation}
with $G$  being the Barnes  $G$-function.
\end{thm}

\begin{rem}
If $n=1$ in \eqref{F-asymptotic-infty2}, we have the large gap asymptotics of the deformed Tracy-Widom distribution
\begin{equation*}
\ln{F}_1(x;\rho)=-\frac{2\beta}{3}|x|^{\frac{3}{2}}+\frac{\beta^2}{4}\ln(8|x|)+
\ln\left[G\left(1+\frac{i\beta}{2}\right)
G\left(1-\frac{i\beta}{2}\right)\right]+O\left(|x|^{-\frac{1}{2}}\right),~~\mathrm{as}~~x\to-\infty,
\end{equation*}
which was first conjectured in
\cite[Conjecture 3]{BCI2016} and lately rigorously proved in \cite{BB2018,BIP2019}.
In the first two non-trivial cases $n=2$ and $n=3$, the formula \eqref{F-asymptotic-infty2} reads
\begin{align*}
\ln{F}_2(x;\rho)&=-\frac{4\beta}{5}|x|^{\frac{5}{4}}-\frac{\beta\tau_1}{3}|x|^{\frac{3}{4}}
-\frac{3\beta\tau_1^2}{8}|x|^{\frac{1}{4}}+\frac{5\beta^2}{16}\ln|x|\nonumber\\
&~~~+\beta^2\ln{2}+\ln\left[G\left(1+\frac{i\beta}{2}\right)
G\left(1-\frac{i\beta}{2}\right)\right]+O\left(|x|^{-\frac{1}{4}}\right),
\end{align*}
and
\begin{align*}
\ln{F}_3(x;\rho)&=-\frac{6\beta}{7}|x|^{\frac{7}{6}}-\frac{\beta\tau_2}{5}|x|^{\frac{5}{6}}
-\frac{\beta}{3}\left(\frac{\tau_2^2}{4}-\tau_1\right)|x|^{\frac{3}{6}}
-\frac{\beta\tau_2}{6}\left(\frac{7\tau_2^2}{36}-\tau_1\right)|x|^{\frac{1}{6}}
+\frac{7\beta^2}{24}\ln|x|\nonumber\\
&~~~+\frac{\beta^2}{4}\ln(24)+\ln\left[G\left(1+\frac{i\beta}{2}\right)
G\left(1-\frac{i\beta}{2}\right)\right]+O\left(|x|^{-\frac{1}{6}}\right).
\end{align*}
As $n$ increases, the formula becomes longer but still can be completely determined through \eqref{ak-expression} and \eqref{F-asymptotic-infty2}. It is worthy to note that the nontrivial constant term $C_0$ in \eqref{C0} is independent of $n$. We also refer to \cite{BDIK2015,BIP2019,Charlier21,CC2021,DXZ22a,DXZ22b,DZ22} and references therein for recent works on large deformations of other distributions arising from random matrix theory and beyond.
\end{rem}

Finally, we give two applications of Theorem \ref{thm:largegap}. In view of the integral representation of $F_n$ given in  \eqref{F1-expression}, we can reformulate \eqref{F-asymptotic-infty2} in terms of the higher order Painlev\'e transcendent, which leads to the following total integral of $q_{n}$.
\begin{corollary}
Under the assumptions of Theorem \ref{q-asymptotic}, we have
\begin{align}\label{total-integral}
&\lim_{x\to-\infty}\left(-\int_x^{+\infty}(s-x)
q_{n}^2\left((-1)^{n+1}s;\rho\right)\mathrm{d}s
+\beta\sum_{k=0}^{n}\frac{2na_k(-x)^{\frac{1+2(n-k)}{2n}}}{1+2(n-k)}
-\frac{(2n+1)\beta^2}{8n}\ln(-x)\right)\nonumber\\
&=\frac{\beta^2}{4}\ln(8n)
+\ln\left[G\left(1+\frac{i\beta}{2}\right)G\left(1-\frac{i\beta}{2}\right)\right].
\end{align}
\end{corollary}

\begin{rem}
The total integral \eqref{total-integral} generalizes the case $n=1$ for Ablowitz-Segur solution of the Painlev\'e II equation, which was first conjectured in
\cite[Conjecture 3]{BCI2016}, and rigorously proved in \cite[Corollary 1.7]{BB2018}.
\end{rem}

As the other application of Theorem \ref{thm:largegap}, we are able to establish asymptotic statistical properties and a rigidity result for the counting function of the higher order Airy point process characterized by
the kernel \eqref{kernal}. For this purpose, we denote by $\lambda_1>\lambda_2>\cdots$ the random points in the process and set $x_j=-\lambda_j$ to be the opposite points. Let $N(x)$ be the random variable that counts the number of opposite points $x_j$ falling  into the interval $(-\infty,x]$, we then have the following result.
\begin{corollary}\label{thm3}
As $x\to+\infty$, we have
\begin{align}\label{expection}
\mathbb{E}(N(x))&=\mu(x)+o(1),\\
\mathrm{Var}(N(x))&=\sigma^2(x)+\frac{\ln(8n)+1+\gamma_{_E}}{2\pi^2}
+o(1),\label{variation}
\end{align}
where $\gamma_{_E}$ is Euler's constant and
\begin{equation}\label{eq:sigma}
\mu(x)=\frac{1}{\pi}\sum_{k=0}^{n}\frac{2na_k}{1+2(n-k)}x^{\frac{1+2(n-k)}{2n}},\qquad \sigma^2(x)=\frac{2n+1}{4n\pi^2}\ln{x}.
\end{equation}
Furthermore, the normalized random variable $(N-\mu)/\sigma$ converges in distribution to the normal law $\mathcal{N}(0,1)$.

Let $-x_1>-x_2>\cdots$ be the random points in the process, then for any $\epsilon>0$, we have
\begin{equation}\label{eq:upperB}
\lim_{k_0\to\infty}\mathrm{P} \left(\sup\limits_{k\geq k_0}\frac{|\mu(x_k)-k|}{ \ln k} \leq  \frac{1}{\pi}+\epsilon\right)=1.\end{equation}
\end{corollary}
\begin{proof}
For $\gamma\geq 0$, it follows from the definition of the random variable $N$ that
\begin{equation}
\mathbb{E}\left(\mathrm{e}^{-2\pi\gamma{N}(x)}\right)=\det\left(I-\left(1-\mathrm{e}^{-2\pi\gamma}\right)
\mathcal{K}_{n,x}\right)=F_n(-x;\rho),
\end{equation}
where $F_n(x;\rho)$ is defined by \eqref{Fred-Deter} with $\rho^2=1-\mathrm{e}^{-2\pi\gamma}$.
As a direct consequence of the  large gap asymptotics \eqref{F-asymptotic-infty2}, we have
\begin{equation}\label{eq: Easy}
\mathbb{E}(\mathrm{e}^{-2\pi\gamma{N}(x)})
=(8n)^{\gamma^2}G(1+i\gamma)G(1-i\gamma)\mathrm{e}^{-2\pi\gamma\mu(x)+2\pi^2\gamma^2\sigma^2(x)}
\left[1+o(1)\right],\qquad x\to+\infty,
\end{equation}
where $\mu$ and $\sigma$ are defined in \eqref{eq:sigma}. Since the Barnes  $G$-function has the approximation
\begin{equation}
G(1+z)=1+\frac{\ln(2\pi)-1}{2}z
+\left[\frac{(\ln(2\pi)-1)^2}{8}-\frac{\gamma_{_E}+1}{2}\right]z^2+O(z^3),
\qquad z\to 0;
\end{equation}
cf. \cite[Equation 5.17.4]{NIST}, the leading term in \eqref{eq: Easy} admits an expansion
\begin{align}\label{Fred-rep-asymp}
& (8n)^{\gamma^2}G(1+i\gamma)G(1-i\gamma)\mathrm{e}^{-2\pi\gamma\mu(x)+2\pi^2\gamma^2\sigma^2(x)}
\nonumber\\
&=1-2\pi\mu(x)\gamma+\left[\ln(8n)+1+\gamma_{_E}+2\pi^2(\mu^2(x)+\sigma^2(x))\right]\gamma^2
+O(\gamma^3),\qquad \gamma\to0.
\end{align}
By comparing the above formula with the expansion
\begin{equation}
\mathbb{E}(\mathrm{e}^{-2\pi\gamma{N}(x)})=1-2\pi\mathbb{E}(N(x))\gamma
+2\pi^2\mathbb{E}(N^2(x))\gamma^2
+O(\gamma^3),\qquad \gamma\to0,
\end{equation}
it is immediate to obtain \eqref{expection} and \eqref{variation} from the estimate \eqref{Fred-a-1} below.

Since  $\sigma(x)\to\infty$ as $x\to +\infty$, we observe from \eqref{eq: Easy} that
\begin{equation}\label{eq: CTL}
\mathbb{E}\left(\mathrm{e}^{s\frac{N(x)-\mu(x)}{\sigma(x)}}\right)\to \mathrm{e}^{\frac{s^2}{2}}, \quad x\to+\infty,
\end{equation}
for any $s\in\mathbb{R}$. This implies the central limit theorem for the random variable $(N-\mu)/\sigma$.

Finally, the upper bound for the maximum fluctuation of the counting function \eqref{eq:upperB} follows directly from asymptotics of the exponential moment \eqref{eq: Easy}  and a rigidity theorem given in \cite[Theorem 1.2]{CC2021}.

This completes the proof of Corollary \ref{thm3}.
\end{proof}

\begin{rem}
Corollary \ref{thm3} extends the relevant results established in \cite{CC2021,Sosa} for the Airy process.
\end{rem}

\subsection{Singular asymptotics of $q_{n}((-1)^{n+1}x;\rho)$ and $\frac{\mathrm{d}}{\mathrm{d}x}\ln F_n(x;\rho)$ for $\rho>1$}
If $\rho>1$, we have the following singular  asymptotics for  $q_{n}((-1)^{n+1}x;\rho)$ and $\frac{\mathrm{d}}{\mathrm{d}x}\ln F_n(x;\rho)$.
\begin{thm}\label{thm1}
Let $q_{n}(x;\rho)$ 
defined through \eqref{dF1-expression}
and $F_n(x;\rho)$ be the determinant defined in \eqref{Fred-Deter}. If $\rho>1$, we have, as $x\to-\infty$,
\begin{align}\label{qAymp-2}
q_{n}\left((-1)^{n+1}x;\rho\right)&=\frac{|x|^{\frac{1}{2n}}}
{\sin\left(\sum_{k=0}^{n+1}\frac{2na_k}{1+2(n-k)}|x|^{\frac{1+2(n-k)}{2n}}
-\frac{2n+1}{2n}\kappa\ln|x|+\varphi\right)}+O\left(|x|^{-\frac{1}{2n}}\right),
\end{align}
and
\begin{align}\label{F-asymptotic-infty}
\frac{\mathrm{d}}{\mathrm{d}x}\ln F_n(x;\rho)&=\Bigg[2\kappa
-\cot\left(\sum_{k=0}^{n+1}\frac{2na_k|x|^{\frac{1+2(n-k)}{2n}}}{1+2(n-k)}
-\frac{2n+1}{2n}\kappa\ln|x|+\varphi\right)
\Bigg]|x|^{\frac{1}{2n}}+O\left(|x|^{-\frac{1}{2n}}\right),
\end{align}
where $a_k$, $k=0,\ldots,n$, are defined in \eqref{ak-expression}, the parameters $\kappa$ and $\varphi$ are related to $\rho$ by
\begin{equation}\label{kappa}
\left\{\begin{aligned}
&\kappa=-\frac{1}{2\pi}\ln(\rho^2-1),\\
&\varphi=-\kappa\ln(8n)+\arg\Gamma\left(\frac{1}{2}+i\kappa\right)
+\frac{\pi}{2}.
\end{aligned}\right.
\end{equation}
The error terms in the asymptotic expansions are uniformly valid for $x$ bounded away from the singularities appearing on the right-hand sides of \eqref{qAymp-2} and \eqref{F-asymptotic-infty}.  \end{thm}

\begin{rem}
In the case that $\tau_1=\cdots=\tau_{n-1}=0$, the error terms in \eqref{qAymp-2} and \eqref{F-asymptotic-infty} can be improved to be
\begin{align}\label{qAymp-2-r}
q_{n}\left((-1)^{n+1}x;\rho\right)&=\frac{|x|^{\frac{1}{2n}}}
{\sin\left(\frac{2n}{2n+1}|x|^{\frac{2n+1}{2n}}
-\frac{2n+1}{2n}\kappa\ln|x|+\varphi\right)}+O\left(|x|^{-1}\right),
\end{align}
and
\begin{align}\label{F-asymptotic-infty-r}
\frac{\mathrm{d}}{\mathrm{d}x}\ln F_n(x;\rho)&=\Bigg[2\kappa
-\cot\left(\frac{2n}{2n+1}|x|^{\frac{2n+1}{2n}}
-\frac{2n+1}{2n}\kappa\ln|x|+\varphi\right)
\Bigg]|x|^{\frac{1}{2n}}+O\left(|x|^{-1}\right).
\end{align}
If $n=1$, the results \eqref{qAymp-2-r} and \eqref{kappa} were first obtained in \cite{Kap} using the isomonodromy method; see also \cite{BI}.
\end{rem}

It is readily seen from the asymptotic formula \eqref{qAymp-2} that $q_{n}\left((-1)^{n+1}x;\rho\right)$ possesses a sequence of real simple poles clustering at negative infinity. As an application, we could approximate the locations of these poles as stated below.

\begin{corollary}\label{cor}
Under the assumptions of Theorem \ref{thm1} and
with
$\tau_1=\cdots=\tau_{n-1}=0$, the function  $q_{n}\left((-1)^{n+1}x;\rho\right)$ has infinitely many simple poles $\{p_m \}_{m\in \mathbb{N}}$ on the real axis, such that
\begin{align}\label{poles}
p_m=&-\left(\frac{2n+1}{2n}\pi\right)^{\frac{2n}{2n+1}}m^{\frac{2n}{2n+1}}
\bigg[1+\frac{2n\kappa}{(2n+1)\pi}\frac{\ln{m}}{m}
\nonumber\\
&\qquad\qquad+\frac{2n}{(2n+1)\pi}\left(\kappa\ln\frac{(2n+1)\pi}{2n}-\varphi\right)
\frac{1}{m}+O\left(\frac{\ln^2m}{m^2}\right)\bigg],\quad m\to\infty,
\end{align}
where $\kappa$ and $\varphi$ are given by \eqref{kappa}.
\end{corollary}
Following the proof of Corollary \ref{cor}, we also have
$$p_m= -\left(\frac{2n+1}{2n}\pi\right)^{\frac{2n}{2n+1}}m^{\frac{2n}{2n+1}}+O\left(m^{\frac{2(n-1)}{2n+1}}\right),$$
for general parameters $\tau_1,\dots,\tau_{n-1}$.  More terms in the expansion can also be obtained but the expression is more complicated
and we do not pursue this here.

The rest of this paper is devoted to the proofs of our main results, which is based on a Deift-Zhou nonlinear steepest descent analysis \cite{Deift,DZ1993,DZ1995} of the associated Riemann-Hilbert (RH) problem. In Section \ref{RHproblem}, we recall the RH problem that characterizes the higher order Painlev\'{e} II transcendents $q_n$ in \eqref{dF1-expression} and the Fredholm determinants \eqref{Fred-Deter}. Asymptotic analysis of the RH problem  for $0<\rho<1$ and  $\rho>1$ are carried out in Sections \ref{RHanalysis1} and \ref{RHanalysis2}, respectively. 
The proofs of our main results are presented in Section \ref{proof-largegap}, as outcomes of the RH analysis. For the convenience of the reader, we also include the parabolic cylinder parametrix used in the analysis in Appendix \ref{PCP}.

\section{RH characterizations of $q_n$ and $F_n$}\label{RHproblem}

The starting point of our analysis is the following RH problem for the Painlev\'{e} II hierarchy \eqref{PIIhierarchy} with specified Stokes multiplier. For more information, we refer to \cite{CCG2021} and \cite[Section 4.2]{CIK2010}.

\subsection*{RH problem for $\Psi$}
\begin{description}
\item{(1)} $\Psi(z):=\Psi(z,x;\rho)$ is analytic for $z\in\mathbb{C}\setminus(\Gamma_+\cup\Gamma_-)$, where $\Gamma_+$ and $\Gamma_-$ are symmetric with respect to the real line, and $\Gamma_+$ is any contour lying in the upper half plane which begins at $\infty\mathrm{e}^{\frac {4n+1}{4n+2}\pi i}$ and ends at $\infty\mathrm{e}^{\frac {1}{4n+2}\pi i}$; see Figure \ref{PIIjump} for an illustration.

\begin{figure}[H]
  \centering
  \includegraphics[width=10cm,height=6cm]{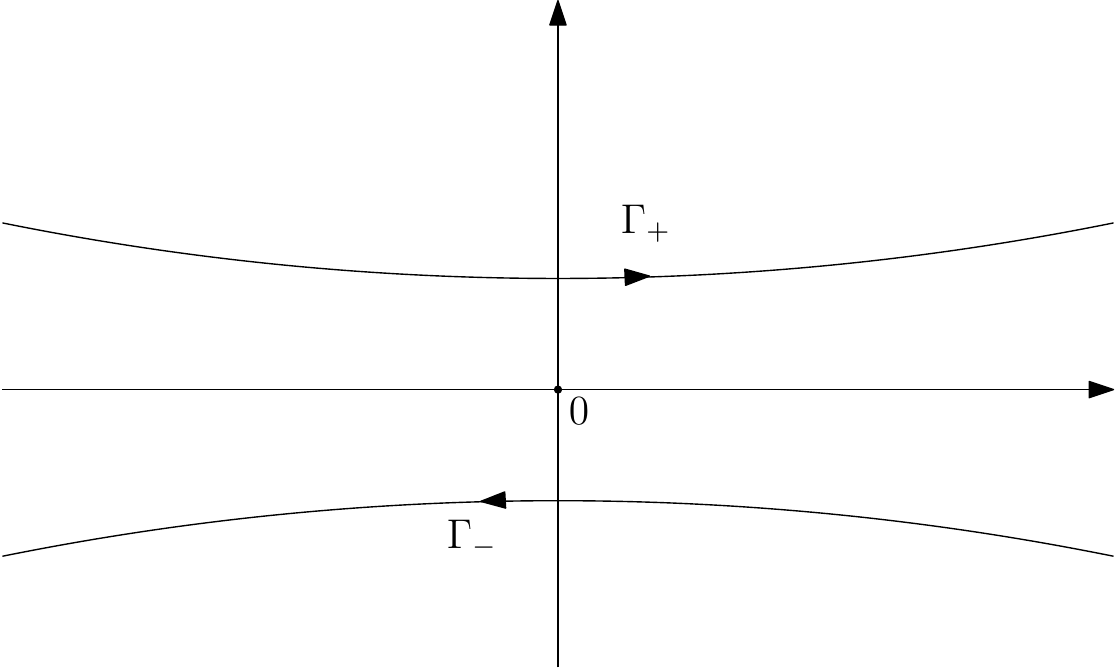}\\ 
  \caption{The jump contours $\Gamma_+$ and $\Gamma_-$ of the RH problem for $\Psi$.
}\label{PIIjump}
\end{figure}

\item{(2)} $\Psi(z)$ satisfies the jump condition
$$\Psi_{+}(z)=\Psi_{-}(z)J_{\Psi}(z),$$
where
\begin{equation}\label{Psi-jump}
J_{\Psi}(z)=\left\{\begin{aligned}
&\begin{pmatrix} 1 & 0 \\ \rho\,\mathrm{e}^{2\theta(z)} & 1\end{pmatrix}, && z\in\Gamma_+,\\
&\begin{pmatrix} 1 & \rho\,\mathrm{e}^{-2\theta(z)} \\ 0 & 1\end{pmatrix}, && z\in\Gamma_-,
\end{aligned}\right.
\end{equation}
with
\begin{align}\label{eq:theta}
\theta(z):&=\theta(z;x,\tau_1,\ldots,\tau_{n-1})=i\frac{(2z)^{2n+1}}{4n+2}
+i\sum^{n-1}_{j=1}\frac{(-1)^{n+j}\tau_j}{4j+2}(2z)^{2j+1}+ixz\nonumber\\
&=\frac{(-1)^n}{2}P_{2n+1}(2\mathrm{e}^{\frac{\pi i}{2}}z)+ixz,
\end{align}
and where $P_{2n+1}$ is defined in \eqref{eq:Pn}.
\item{(3)} As $z\to \infty$, we have
\begin{equation}\label{Asyatinfty}
\Psi(z)=I+\frac{\Psi_{1}(x)}{z}
+O\left(z^{-2}\right)
\end{equation}
for some function $\Psi_{1}$.
\end{description}

Recall that $q_{n}$ is the solution of the $n$-th member of the Painlev\'e II hierarchy \eqref{PIIhierarchy} determined through \eqref{dF1-expression}. It is related to the above RH problem via the formulas
\begin{equation}\label{q-expression}
q_{n}(x;\rho)= \begin{cases}2 i\left(\Psi_{1}(x)\right)_{12}=-2 i\left(\Psi_{1}(x)\right)_{21}, & \text{if}~n~\text{is~odd}, \\
2 i\left(\Psi_{1}(-x)\right)_{12}=-2 i\left(\Psi_{1}(-x)\right)_{21}, & \text{if}~n~ \text{is~even}.\end{cases}
\end{equation}
where $\Psi_1$ is the coefficient of $1/z$  in \eqref{Asyatinfty} and $A_{ij}$ stands for the $(i,j)$-th entry of a matrix $A$; see \cite[Equation (2.27)]{CCG2021}. Note that $\Psi$ satisfies the symmetry relation
\begin{equation}\label{eq:PsiSymm}
\sigma_3\Psi(z,x;\rho)\sigma_3=\Psi(z,x;-\rho),\qquad \sigma_3=\begin{pmatrix} 1 & 0\\ 0 & -1 \end{pmatrix},
\end{equation}
it is then easily seen from \eqref{q-expression} that $q_{n}(x;-\rho)=-q_{n}(x;\rho)$, as claimed in \eqref{symmetry}.


We also have the following connections between the Fredholm determinant $F_n$ defined in \eqref{Fred-Deter} and the RH problem for $\Psi$.
\begin{pro}
The Fredholm determinant $F_n(x;\rho)$ satisfies the following differential identities:
\begin{align}\label{Fred-expression-to-x}
\partial_x\ln{F_n(x;\rho)}&=2i(\Psi_1(x))_{11},
\\
\label{Fred-expression-to-rho}
\partial_{\rho}\ln{F_n(x;\rho)}&=\frac{1}{\rho}\int_{\Gamma_+\cup\Gamma_-}
\mathrm{Tr}\left[\Psi^{-1}(z)\Psi'(z)(J_{\Psi}(z)-I)\right]\frac{\mathrm{d}z}{2\pi i},
\end{align}
where $\Psi_1$ is defined in \eqref{Asyatinfty} and $J_{\Psi}$ is given in \eqref{Psi-jump}.
\end{pro}
\begin{proof}
The identity \eqref{Fred-expression-to-x} follows directly from \cite[Equations (2.19) and (2.25)]{CCG2021}. To show \eqref{Fred-expression-to-rho}, we see from the results in \cite[Section 5.1]{Ber} that
\begin{equation}\label{eq:BCdiff}
\partial_{\rho}\ln{F_n(x;\rho)}=\int_{\Gamma_+\cup\Gamma_-}\mathrm{Tr}
\left[\Psi^{-1}_-(z)\Psi'_-(z)\partial_{\rho}(J_{\Psi}(z))J_{\Psi}^{-1}(z)\right]
\frac{\mathrm{d}z}{2\pi i}.
\end{equation}
In view of $J_{\Psi}$ in \eqref{Psi-jump}, it is readily seen that
$$
\partial_{\rho}(J_{\Psi}(z))J_{\Psi}^{-1}(z)=\frac{1}{\rho}(J_{\Psi}(z)-I),
$$
which implies
$$
\partial_{\rho}\ln{F_n(x;\rho)}=\frac{1}{\rho}\int_{\Gamma_+\cup\Gamma_-}
\mathrm{Tr}\left[\Psi_-^{-1}(z)\Psi_-'(z)(J_{\Psi}(z)-I)\right]\frac{\mathrm{d}z}{2\pi i}.
$$
We could remove the subscript $_{-}$ in the above formula by noting that
\begin{equation*}
\mathrm{Tr}\left[\Psi_-^{-1}(z)\Psi_-'(z)(J_{\Psi}(z)-I)\right]=\mathrm{Tr}\left[\Psi_+^{-1}(z)\Psi_+'(z)(J_{\Psi}(z)-I)\right], \qquad z\in \Gamma_+\cup\Gamma_-,
\end{equation*}
which leads to \eqref{Fred-expression-to-rho}.
\end{proof}

\section{Asymptotic analysis of the RH problem for $\Psi$ with $0<\rho<1$}\label{RHanalysis1}
In this section, we  will perform a Deift-Zhou nonlinear steepest descent analysis to the RH problem for $\Psi$ as $x\to-\infty$ with the parameter $0<\rho<1$. It consists of a series of explicit and invertible transformations which leads
to an RH problem tending to the identity matrix for large negative $x$.

\subsection{First transformation: $\Psi \to X$}
The first transformation is a rescaling, which is defined by
\begin{equation}\label{rescaling}
X(z)=\Psi\left(|x|^{\frac{1}{2n}}z\right), \qquad x<0.
\end{equation}
Thus, $X$ satisfies the following RH problem.

\subsection*{RH problem for $X$}
\begin{description}
\item{(1)} $X(z)$ is analytic for $z\in\mathbb{C}\setminus(\gamma_+\cup\gamma_-)$, where $\gamma_{\pm}:=|x|^{-\frac{1}{2n}}\Gamma_{\pm}$ with $\Gamma_{\pm}$ depicted in Figure \ref{PIIjump}.
\item{(2)} $X(z)$ satisfies the jump condition
$$X_+(z)=X_-(z)J_{X}(z),$$
where
\begin{equation}
J_{X}(z)=\left\{\begin{aligned}
&\begin{pmatrix} 1 & 0 \\ \rho\,\mathrm{e}^{2tg(z)} & 1\end{pmatrix}, && z\in\gamma_+,\\
&\begin{pmatrix} 1 & \rho\,\mathrm{e}^{-2tg(z)} \\ 0 & 1\end{pmatrix}, &&
z\in\gamma_-,
\end{aligned}\right.
\end{equation}
and where $t=|x|^{\frac{2n+1}{2n}}$,
\begin{equation}\label{g-expression}
g(z):=g(z;x,\tau_1,\ldots,\tau_{n-1})
=i\frac{(2z)^{2n+1}}{4n+2}-iz+i\sum_{j=1}^{n-1}\frac{(-1)^{n+j} \tau_j}{4j+2}(2z)^{2 j+1}|x|^{\frac{j-n}{n}}.
\end{equation}
\item{(3)} As $z\to\infty$, we have
\begin{equation}
X(z)=I+\frac{\Psi_{1}(x)}{|x|^{\frac{1}{2n}}z}
+O\left(z^{-2}\right),
\end{equation}
where $\Psi_1$ is defined in \eqref{Asyatinfty}.
\end{description}

\subsection{Second transformation: $X \to Y$}
The second transformation involves a deformation of the jump contour of the RH problem for $X$.
To proceed, we begin with studies of the saddle points of the phase function $g$ in \eqref{g-expression}.
\begin{lemma}
For sufficiently large negative $x$, $g(z)$ has exactly two real saddle points $z_+>0$ and $z_-=-z_+$. Moreover, $z_+$ admits the following expansion
\begin{equation}\label{z+-expansion}
z_+=\frac{1}{2}+\frac{1}{2}\sum^{\infty}_{k=1}a_k|x|^{-\frac{k}{n}},
\qquad    x \to -\infty,
\end{equation}
where the coefficients $a_k$, $k\in\mathbb{N}$, are given in \eqref{ak-expression}.
\end{lemma}

\begin{proof}
From \eqref{g-expression}, it follows that
\begin{equation}\label{eq:dg} 
|x|g'(z)=i\left(Q\left(2z|x|^{\frac{1}{2n}}\right)-|x|\right), \qquad x<0,
\end{equation}
where $Q$ is the even polynomial defined in \eqref{ell}.
Thus, we have $ g'(z)=0 $ if and only if  $Q\left(2z|x|^{\frac{1}{2n}}\right)=|x|$.
Since $Q(z)$ is a polynomial with real coefficients and behaves like $z^{2n}$ as $z\to\infty$, there exists a constant $M>0$  such that $Q(z)$ is strictly increasing on $[M,\infty)$.  Suppose $|x|$ is large enough such that $Q(z)<|x|$ for all $z\in[0,M]$, it is then easily seen that the equation $Q(z) = |x|$ has a unique positive solution $z_0\in[M,\infty)$.
Furthermore, according to \cite[Lemma 4.3]{CCG2021}, the solution $z_0$ has the following fractional power series expansion
 \begin{equation}\label{eq:z0Expansion}
z_0=|x|^{\frac{1}{2n}}+\sum^{\infty}_{k=0}b_k |x|^{-\frac{k}{2n}}, \qquad x\to -\infty,
\end{equation}
where $b_0=0$ and
 \begin{equation}\label{eq:bk}
b_k=\frac{1}{k}\Res\limits_{z=\infty}Q^{\frac{k}{2n}}(z), \qquad k\in\mathbb{N}.
\end{equation}
Since $Q$ is an even polynomial, we have
 \begin{equation}\label{eq:bm0}
b_{2m}=0, \quad m\in\mathbb{N}, \qquad \textrm{and} \qquad Q(-z_0)=|x|.
\end{equation}
Hence, for sufficiently large negative $x$, there exist exactly two real saddle points $z_{\pm}$ of $g(z)$, which are symmetric with respect to the origin and $$z_+=\frac{1}{2}|x|^{-\frac{1}{2n}}z_0.$$
Inserting  \eqref{eq:z0Expansion}--\eqref{eq:bm0} into the above formula gives us \eqref{z+-expansion}.
\end{proof}

\begin{figure}[t]
  \centering
  \includegraphics[width=12cm,height=4.5cm]{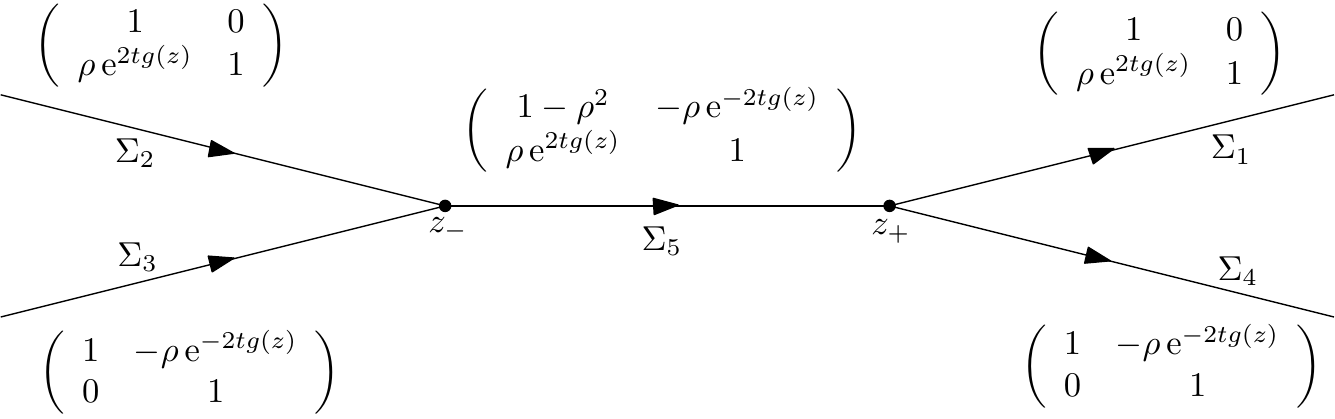}\\
\caption{The jump contours $\Sigma_i$, $i=1,\ldots,5$, and jump matrix $J_Y$ of the RH problem for $Y$.}\label{GammaS}
\end{figure}

We next deform the jump contours $\gamma_+$ and $\gamma_-$ for $X$ into five rays $\Sigma_i$, $i=1,\ldots,5$. Here, the rays $\Sigma_{1,4}$ and $\Sigma_{2,3}$ are emanating from the two saddle points $z_{\pm}$ of $g$, respectively, and $\Sigma_5=[z_-,z_+]$; see Figure \ref{GammaS} for an illustration. We choose the rays $\Sigma_i$, $i=1,\ldots,4$, such that
\begin{equation}\label{eq:ginequa}
\left\{
         \begin{array}{ll}
           \Re g(z)<0, & \hbox{$z\in \Sigma_1 \cup \Sigma_2$,} \\
           \Re g(z)>0, & \hbox{$z\in \Sigma_3 \cup \Sigma_4$.}
         \end{array}
       \right.
\end{equation}
The possibility of this choice follows from the fact that $g=ig_0+O(|x|^{-\frac{1}{n}})$ as $x\to-\infty$ with
\begin{equation}\label{g0g1-expression}
g_0(z):=\frac{(2z)^{2 n+1}}{4n+2}-z,
\end{equation}
and the signature of $\Re{ig}_0$ as depicted in Figure \ref{gsign}. The second transformation $X\to{Y}$ is then defined by
\begin{equation}\label{Y-transformation}
Y(z)=X(z)\left\{\begin{aligned}
&\begin{pmatrix} 1 & 0 \\ \rho\,\mathrm{e}^{2tg(z)} & 1\end{pmatrix}, && \textrm{$z$ in the region bounded by $\gamma_+$ and $\Sigma_2\cup\Sigma_5\cup\Sigma_1$,}\\
&\begin{pmatrix} 1 & \rho\,\mathrm{e}^{-2tg(z)} \\ 0 & 1\end{pmatrix}, &&
\textrm{$z$ in the region bounded by $\gamma_-$ and $\Sigma_3\cup\Sigma_5\cup\Sigma_4$,}
\\
&I , &&
\textrm{elsewhere.}
\end{aligned}\right.
\end{equation}


In view of the RH problem for $X$, it is direct to check that $Y$ satisfies the following RH problem.

\subsection*{RH problem for $Y$}
\begin{description}
\item{(1)} $Y(z)$ is analytic for $z\in\mathbb{C}\setminus(\cup_1^5\Sigma_k)$, where the rays $\Sigma_{k}$, $k=1,\cdots,5$, are depicted in Figure \ref{GammaS}.

\item{(2)} $Y(z)$ satisfies the jump condition
$$Y_{+}(z)=Y_{-}(z)J_{Y}(z),$$
where the jump matrix $J_Y$ is shown in Figure \ref{GammaS}.
\item{(3)} As $z\to\infty$, we have
\begin{equation}
Y(z)=I+\frac{\Psi_{1}(x)}{|x|^{\frac{1}{2n}}z}
+O\left(z^{-2}\right),
\end{equation}
where $\Psi_1$ is defined in \eqref{Asyatinfty}.
\end{description}
\begin{figure}[t]
  \centering
  \includegraphics[width=7.2cm,height=7cm]{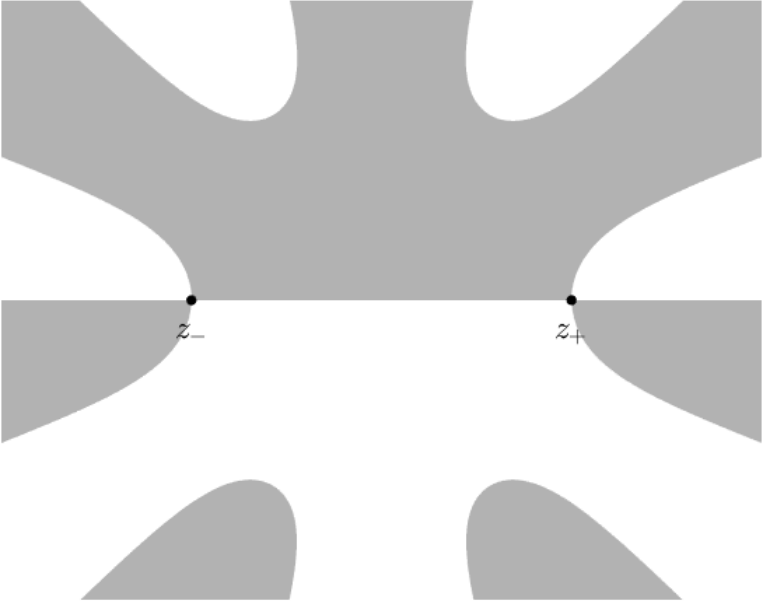}
  \hspace{0.6cm}
  \includegraphics[width=7.2cm,height=7cm]{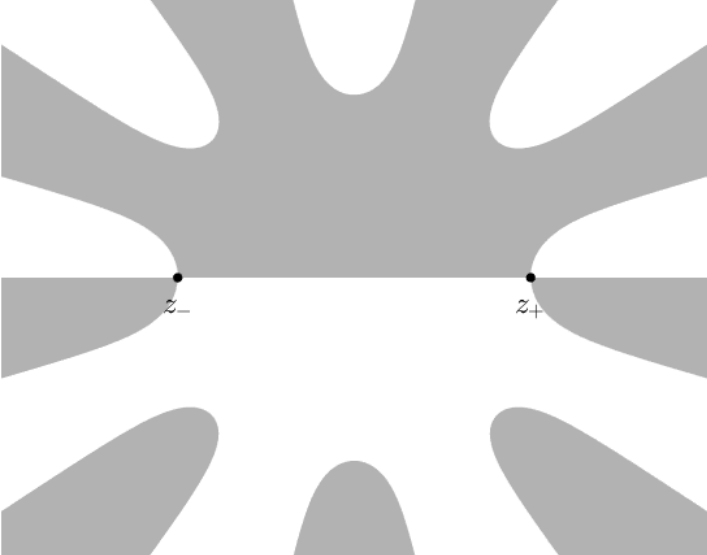}
  \caption{The signatures of $\Re ig_0$ in the cases $n=3$ (left) and $n=4$ (right). The shaded areas indicate the regions where $\Re ig_0>0$, while the white areas indicate the regions where $\Re ig_0<0$.This case corresponds to $\tau_1=\cdots=\tau_{n-1}=0$ in \eqref{g-expression}, thus $z_{\pm}=\pm\frac{1}{2}$. }\label{gsign}
\end{figure}

We conclude this subsection by showing the expansion of $g(z_+)$ as $x\to -\infty$ for later use.
\begin{lemma}\label{eq:gz}
Let $z_+$ be the positive saddle point of $g$. We have, as $x\to -\infty$,
\begin{equation}\label{eq:gexpansion}
2ig(z_+)|x|^{\frac{2n+1}{2n}}=\sum_{k=0}^{\infty}\frac{2na_k}{1+2(n-k)}
|x|^{\frac{1+2(n-k)}{2n}},
\end{equation}
where $a_k$, $k\in \{0\}\cup\mathbb{N}$, are given in \eqref{ak-expression}.
\end{lemma}
\begin{proof}
From \eqref{eq:theta} and \eqref{g-expression}, it follows that
\begin{equation}\label{eq:gtheta}
g(z)|x|^{\frac{2n+1}{2n}}=\theta\left(|x|^{\frac{1}{2n}}z;x\right), \qquad x<0.
\end{equation}
Since $g'(z_+)=0$, we have $\theta_z\left(|x|^{\frac{1}{2n}}z_+;x\right)=0$, where $\theta_z$ denotes the partial derivative of $\theta(z;x)$ with respect to $z$.
Differentiating both sides of \eqref{eq:gtheta} with respect to $x$, we have
\begin{align} \nonumber
\frac{\mathrm{d}}{\mathrm{d}x} \left(g(z_+)|x|^{\frac{2n+1}{2n}}\right)
&=\frac{\mathrm{d}}{\mathrm{d}x}\theta\left(|x|^{\frac{1}{2n}}z_+ ;x\right)
=\theta_z\left(|x|^{\frac{1}{2n}}z_+\right)\frac{\mathrm{d}}{\mathrm{d}x}
\left(|x|^{\frac{1}{2n}}z_+ \right)+\theta_x\left(|x|^{\frac{1}{2n}}z_+;x \right)\nonumber \\
&=\theta_x\left(|x|^{\frac{1}{2n}}z_+\right) \label{z0-equation}
=i|x|^{\frac{1}{2n}}z_+.\end{align}
Integrating both sides of \eqref{z0-equation} and combining with the expansion \eqref{z+-expansion}, we obtain
\begin{equation}\label{eq:gexpansion1}
2ig(z_+)|x|^{\frac{2n+1}{2n}}=\sum_{k=0}^{\infty}\frac{2na_k}{1+2(n-k)}|x|^{\frac{1+2(n-k)}{2n}}
+C_1,
\end{equation}
where $C_1$ is a constant of integration. To evaluate $C_1$, we observe that, by substituting the expansion \eqref{z+-expansion} of $z_+$ into the polynomial \eqref{g-expression}, $g(z_+)$ admits an expansion of the form
\begin{equation}\label{eq:gexpansion2}
g(z_+)=\sum_{k=0}^{\infty}d_k |x|^{-\frac{k}{n}}.
\end{equation}
Comparing \eqref{eq:gexpansion1} with \eqref{eq:gexpansion2} leads to $C_1=0$.
\end{proof}

\subsection{Third transformation: $Y \to T$}
Due to the inequalities \eqref{eq:ginequa}, it is immediate that the jump matrix $J_{Y}(z)$ tends to the identity matrix exponentially fast as $x\to-\infty$, except for $z \in (z_-,z_+)$. Moreover, $J_{Y}$ is highly oscillatory for large $|x|$ along $(z_-,z_+)$.
Using again  the signature of $\Re ig_0$ as illustrated in Figure \ref{gsign}, we could turn the high oscillations into exponential decays by opening lens around the segment $(z_-,z_+)$ as shown in Figure \ref{GammaT} in such a way that $\Re g(z)>0$ for $z\in \Sigma_6$ and $\Re g(z)<0$ for $z\in \Sigma_7$. Based on the factorization
$$
J_Y(z)=\begin{pmatrix}1 & 0 \\ \frac{\rho}{1-\rho^{2}} \mathrm{e}^{2tg(z)} & 1 \end{pmatrix}
(1-\rho^2)^{\sigma_3}
\begin{pmatrix} 1 & -\frac{\rho}{1-\rho^{2}} \mathrm{e}^{-2tg(z)} \\ 0 & 1 \end{pmatrix}, \qquad z\in(z_-,z_+),
$$
the third transformation $Y\to{T}$  is then defined by
\begin{equation}\label{T-transformation}
T(z)=Y(z)H(z),
\end{equation}
where
\begin{equation}\label{H-expression}
H(z)=\left\{\begin{aligned}
&\begin{pmatrix}1 & \frac{\rho}{1-\rho^{2}} \mathrm{e}^{-2tg(z)} \\ 0 & 1 \end{pmatrix}, &&\mathrm{for}~z~\mathrm{in}~\mathrm{the}~\mathrm{upper}~\mathrm{lens}~\mathrm{region},\\
&\begin{pmatrix}1 & 0 \\ \frac{\rho}{1-\rho^{2}} \mathrm{e}^{2tg(z)} & 1 \end{pmatrix},
&&\mathrm{for}~z~\mathrm{in}~\mathrm{the}~\mathrm{lower}~\mathrm{lens}~\mathrm{region},\\
&I, && \mathrm{elsewhere}.
\end{aligned}\right.
\end{equation}
It is then straightforward to check that ${T}$ solves the following RH problem.

\subsection*{RH problem for $T$}
\begin{description}
\item{(1)} ${T}(z)$ is analytic for $z\in\mathbb{C}\setminus(\cup_1^7\Sigma_k)$, where $\Sigma_{k}$, $k=1,\ldots,7$, are depicted in Figure \ref{GammaT}.

\item{(2)} $T(z)$ satisfies the jump condition $${T}_{+}(z)={T}_{-}(z)J_{T}(z),$$
   where $J_{T}$ is shown in Figure \ref{GammaT}.
\item{(3)} As $z\to\infty$, we have
\begin{equation}
T(z)=I+\frac{\Psi_{1}(x)}{|x|^{\frac{1}{2n}}z}
+O\left(z^{-2}\right),
\end{equation}
where $\Psi_1$ is defined in \eqref{Asyatinfty}.
\end{description}

\begin{figure}[t]
  \centering
  \includegraphics[width=12cm,height=4.5cm]{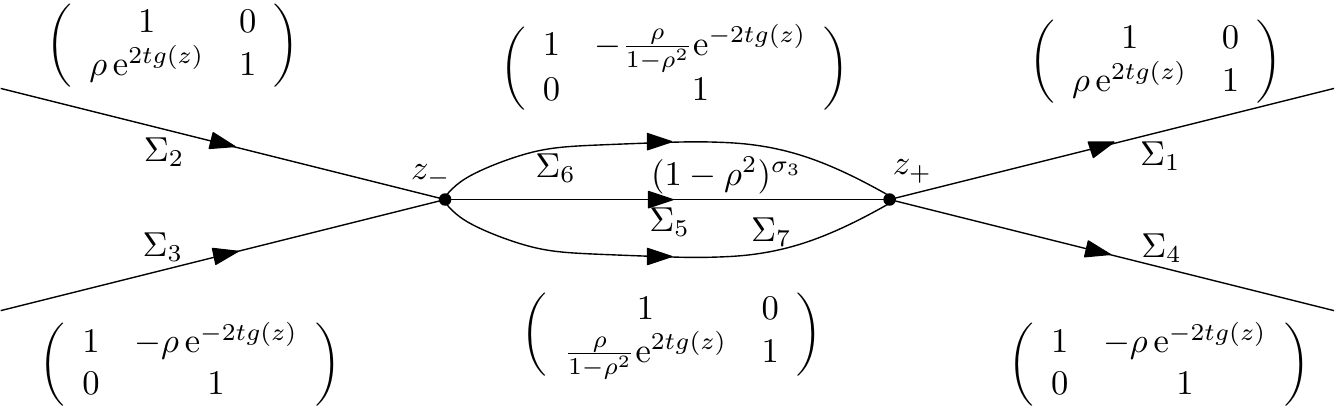}\\
  \caption{The jump contours $\Sigma_k$, $k=1,\ldots,7$, and jump matrix $J_T$ of the RH problem for $T$.}\label{GammaT}
\end{figure}

By \eqref{g-expression}, we observe that $\mathrm{e}^{-2tg(-z)} = \mathrm{e}^{2tg(z)}$, which leads to the symmetry relation
\begin{equation}\label{T-symmetry}
\sigma_1T(-z)\sigma_1=T(z),\qquad \sigma_1=\begin{pmatrix} 0 & 1 \\ 1 & 0 \end{pmatrix}.
\end{equation}

\subsection{Global parametrix}\label{subsec:global1}
As $x\to -\infty$, $J_{T}\to I$ exponentially fast except for $z\in (z_-,z_+)$. It is then natural to consider the following
RH problem for the global parametrix $P^{(\infty)}$.
\subsection*{RH problem for $P^{(\infty)}$}
\begin{description}
  \item{(1)} $P^{(\infty)}(z)$ is analytic for $z\in \mathbb{C}\setminus [z_{-},z_{+}]$.
  \item{(2)} $P^{(\infty)}(z)$ satisfies the jump condition
  \begin{equation}\label{Pinftyjump}
  {P}^{(\infty)}_{+}(z)={P}^{(\infty)}_{-}(z)
  \begin{pmatrix}1-\rho^2 & 0 \\ 0 & (1-\rho^2)^{-1}\end{pmatrix},\qquad z\in (z_-,z_+).
  \end{equation}
  \item{(3)} As $z\rightarrow \infty$, we have ${P}^{(\infty)}(z)=I+O(z^{-1})$.
\end{description}

The solution to this RH problem is explicitly given by
\begin{equation}\label{Pinfty}
{P}^{(\infty)}(z)=
\left(\frac{z-z_-}{z-z_+}\right)^{\nu\sigma_3},\qquad  \arg(z-z_\pm)\in(-\pi,\pi),
\end{equation}
where
\begin{equation}\label{nu}
\nu:=-\frac{1}{2\pi{i}}\ln(1-\rho^2).
\end{equation}
Since $0<\rho<1$, we have $\nu\in{i\mathbb{R}}$. A direct computation shows that
\begin{equation}\label{Pinfty-expansion}
P^{(\infty)}(z)=I+\frac{2z_+\nu\sigma_3}{z}+O(z^{-2}),\qquad z\to\infty.
\end{equation}

\subsection{Local parametrices at $z_{\pm}$}\label{subsec:localpara}
Since the convergence of $J_T$ to the identity matrix is not uniform near $z=z_\pm$, we have to construct two local parametrices ${P}^{(\pm)}$ in the neighborhoods $\mathcal{D}(z_{\pm})=\{z\in\mathbb{C}\mid|z-z_{\pm}|<\delta\}$ of the saddle points $z_{\pm}$, respectively. They
read as follows.

\subsection*{RH problems for ${P}^{(+)}$ and ${P}^{(-)}$}
\begin{description}
\item{(1)}
${P}^{(+)}(z)$ and ${P}^{(-)}(z)$ are analytic for $z \in \mathcal{D}(z_{+}) \setminus (\cup_{k=1,4,5,6,7} \Sigma_{k})$ and $z \in \mathcal{D}(z_-) \setminus (\cup_{k=2,3,5,6,7} \Sigma_{k})$, respectively. 
\item{(2)} ${P}^{(+)}(z)$ and ${P}^{(-)}(z)$ satisfy the jump conditions
\begin{align}
{P}^{(+)}_+(z) & = {P}^{(+)}_-(z)J_T(z), \qquad z \in \mathcal{D}(z_{+}) \cap (\cup_{k=1,4,5,6,7}\Sigma_{k}),
\label{eq:P+jump}\\
{P}^{(-)}_+(z) & = {P}^{(-)}_-(z)J_T(z), \qquad z \in \mathcal{D}(z_{-}) \cap (\cup_{k=2,3,5,6,7}\Sigma_{k}),
\end{align}
respectively, where $J_{T}$ is shown in Figure \ref{GammaT}.
\item{(3)} As $t=|x|^{\frac{2n+1}{2n}}\to+\infty$,  ${P}^{(\pm)}(z)$ satisfies the  matching condition
\begin{equation}\label{matching}
{P}^{(\pm)}(z)=\left(I+O\left(t^{-\frac{1}{2}}\right)\right){P}^{(\infty)}(z), \qquad z\in \partial \mathcal{D}(z_{\pm}),
\end{equation}
where ${P}^{(\infty)}$ is given in \eqref{Pinfty}.
\end{description}

From the symmetry relation \eqref{T-symmetry}, it is easily seen that
\begin{equation}\label{P-}
{P}^{(-)}(z)=\sigma_1{P}^{(+)}(-z)\sigma_1.
\end{equation}
Thus, we only need to solve the RH problem for ${P}^{(+)}$.

To proceed, we introduce a function
\begin{equation}\label{zeta}
\zeta(z):=2\sqrt{g(z_+)-g(z)}=2\sqrt{ig_0(z_+)-ig_0(z)}+O\left(|x|^{-\frac{1}{n}}\right),
\end{equation}
where $g_0$ is defined in \eqref{g0g1-expression} and the branches of the square roots are taken such that both $\arg(g(z_+)-g(z))$ and $\arg(ig_0(z_+)-ig_0(z))$ belong to $(0,2\pi)$. The function $\zeta$ defines a conformal mapping from $\mathcal{D}(z_+)$ to a neighborhood of $\zeta=0$ for sufficiently large $|x|$. From \eqref{g-expression}, \eqref{z+-expansion} and \eqref{zeta}, it is readily seen that, as $x\to-\infty$,
\begin{equation}\label{zeta'behavior}
\zeta'(z_+)=\mathrm{e}^{\frac{3\pi i}{4}}2\sqrt{2n}\left[1+O\left(|x|^{-\frac{1}{n}}\right)\right],
\end{equation}
and
\begin{equation}\label{zeta''behavior}
\zeta''(z_+)=\mathrm{e}^{\frac{3\pi i}{4}}\frac{4\sqrt{2n}}{3}(2n-1)\left[1+O\left(|x|^{-\frac{1}{n}}\right)\right].
\end{equation}

Let $\Phi_{(\mathrm{PC})}$ be the parabolic cylinder parametrix introduced in Appendix \ref{PCP} with parameter $\nu$ given in \eqref{nu}. We then define
\begin{equation}\label{P+}
{P}^{(+)}(z)=E(z)\Phi_{(\mathrm{PC})}\left(\sqrt{t}\zeta(z)\right)
\left(\frac{h_1}{\rho}\right)^{\frac{\sigma_3}{2}}\sigma_3\mathrm{e}^{tg(z)\sigma_3},
\end{equation}
where $h_1=\sqrt{2\pi}\text{e}^{i\pi\nu}/\Gamma(-\nu)$ and
\begin{equation}\label{E-expression}
E(z)={P}^{(\infty)}(z)\sigma_3
\left(\frac{\rho}{h_1}\right)^{\frac{\sigma_3}{2}}
\left(\sqrt{t}\zeta(z)\right)^{\nu\sigma_3}\mathrm{e}^{-tg(z_+)\sigma_3}.
\end{equation}
From \eqref{Pinfty} and \eqref{zeta}, it follows that $E$ is analytic in $\mathcal{D}(z_+)$. This, together with the jump condition \eqref{PCjumpmatrices} for $\Phi_{(\mathrm{PC})}$, implies that ${P}^{(+)}$ defined in \eqref{P+} indeed satisfies \eqref{eq:P+jump}.
Moreover, we obtain from the large-$\eta$ behavior of $\Phi_{(\mathrm{PC})}$ given in \eqref{PCAsyatinfty} that, as $t\to+\infty$,
\begin{align}\label{P+Pinfty}
P^{(+)}(z)P^{(\infty)}(z)^{-1}&=I+\begin{pmatrix} 0 & -\frac{\nu\rho{t}^{\nu}\zeta(z)^{2\nu-1}(z-z_-)^{2\nu}}
{h_1\mathrm{e}^{2tg(z_+)}(z-z_+)^{2\nu}} \\ -\frac{h_1\mathrm{e}^{2tg(z_+)}(z-z_+)^{2\nu}}
{\rho{t}^{\nu}\zeta(z)^{2\nu+1}(z-z_-)^{2\nu}} & 0 \end{pmatrix}t^{-\frac{1}{2}}+O\left(t^{-1}\right)\nonumber\\
&=I+O\left(t^{-\frac{1}{2}}\right), \qquad z\in \partial\mathcal{D}(z_+),
\end{align}
where we have also made use of the facts that $g(z_+)$ and $\nu$ are purely imaginary; see \eqref{g-expression} and \eqref{nu}.

In conclusion, we have shown that ${P}^{(+)}$ and ${P}^{(-)}$ defined in \eqref{P+} and \eqref{P-} solve the RH problems for ${P}^{(+)}$ and ${P}^{(-)}$, respectively.


\subsection{Small norm RH problem}\label{subsec:final}
The final transformation is defined by
\begin{equation}\label{Final-transformation}
{R}(z)=\left\{\begin{aligned}
&{T}(z){P}^{(\pm)}(z)^{-1},\quad && z\in \mathcal{D}(z_{\pm})\setminus(\cup_{k=1}^7\Sigma_k),\\
&{T}(z){P}^{(\infty)}(z)^{-1},\quad  &&\mathrm{elsewhere}. \end{aligned}
\right.
\end{equation}
We see from the RH problems for $T,{P}^{(\infty)}$ and ${P}^{(\pm)}$ that $R$ satisfies the following RH problem.

\begin{figure}[t]
  \centering
  \includegraphics[width=13cm,height=2.3cm]{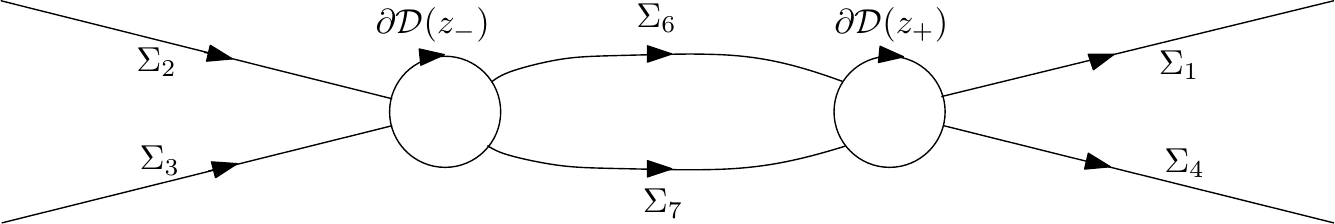}\\ 
  \caption{The jump contour $\Sigma_R$ of the RH problems for ${R}$ and $\widetilde R$.}\label{GammaR}
\end{figure}

\subsection*{RH problem for ${R}$}
\begin{description}
\item{(1)} ${R}(z)$ is analytic for $z\in \mathbb{C}\setminus\Sigma_{R}$,
where
\begin{equation}\label{def:SigmaR}
\Sigma_R:=\cup_{k=1}^7\Sigma_k \cup \partial \mathcal{D}(z_+) \cup \partial \mathcal{D}(z_-) \setminus (\Sigma_5\cup\mathcal{D}(z_+) \cup \mathcal{D}(z_-));
\end{equation}
see Figure \ref{GammaR} for an illustration.
\item{(2)} $R(z)$ satisfies the jump condition
\begin{equation}\label{eq:Rjump}
{R}_+(z)={R}_-(z)J_{R}(z),
\end{equation}
where
 \begin{equation}\label{Rjump}
 J_{R}(z)=\left\{\begin{aligned}
&{P}^{(\pm)}(z){P}^{(\infty)}(z)^{-1},\quad && z\in \partial\mathcal{D}(z_{\pm}),\\
&{P}^{(\infty)}(z)J_{T}(z){P}^{(\infty)}(z)^{-1},\quad &&\textrm{$z \in \Sigma_R\setminus (\partial \mathcal{D}(z_+) \cup \partial \mathcal{D}(z_-))$}.\\ \end{aligned}
\right.
 \end{equation}
\item{(3)} As $z\rightarrow\infty$, we have
\begin{equation}\label{R-expansion}
R(z)=I+\frac{R_1}{z}+O(z^{-2}),
\end{equation}
where $R_1$ is independent of $z$.
\end{description}

Since $P^{(\infty)}$ is independent of $x$ and is uniformly bounded outside $\mathcal{D}(z_+)\cup \mathcal{D}(z_-)$, we have that, as $x\to-\infty$,
$$
J_R(z)=I+O\left(\text{e}^{-c_1|x|^{\frac{2n+1}{2n}}}\right), \qquad z\in \Sigma_R\backslash\left(\partial\mathcal{D}(z_+)\cup\partial \mathcal{D}(z_-)\right),
$$
where $c_1$ is a positive constant. This, together with \eqref{P+Pinfty} and \eqref{Rjump}, implies that $J_R(z)$ has an expansion in the form
\begin{equation}\label{Rjump-estimate}
J_R(z)=I+\frac{J_R^{(1)}(z)}{|x|^{\frac{2n+1}{4n}}}
+\frac{J_R^{(2)}(z)}{|x|^{\frac{2n+1}{2n}}}+O\left(|x|^{-\frac{3(2n+1)}{4n}}\right), \qquad x\to-\infty,
\end{equation}
where
\begin{equation}\label{Rjump12}
J_R^{(j)}(z)=\left\{\begin{aligned}
&E(z)\frac{\Phi_{(\mathrm{PC}), j}}{\zeta(z)^j} E(z)^{-1}, && z\in\partial\mathcal{D}(z_+), \\
&\sigma_1E(-z)\frac{\Phi_{(\mathrm{PC}), j}}{\zeta(-z)^j} E(-z)^{-1}\sigma_1, && z\in\partial\mathcal{D}(z_-),
\end{aligned}\right.
\end{equation}
with $E$ given in \eqref{E-expression} and
$$
\Phi_{(\mathrm{PC}), 1}=\begin{pmatrix} 0 & \nu \\ 1 & 0\end{pmatrix},\qquad
\Phi_{(\mathrm{PC}), 2}=\begin{pmatrix} \frac{\nu(\nu+1)}{2} & 0 \\ 0 & -\frac{\nu(\nu-1)}{2}\end{pmatrix}.
$$
Since the RH problem for $R$ is equivalent to the following singular integral equation
\begin{equation}
R(z)=I+\frac{1}{2\pi i}\int_{\Sigma_R}\frac{R_{-}(s)\left(J_R(s)-I\right)}{s-z}\mathrm{d}s,
\end{equation}
we obtain from a standard theory  (cf. \cite{Deift}) that, as $x\to -\infty$,
\begin{align}
R(z)&=I+\frac{R^{(1)}(z)}{|x|^{\frac{2n+1}{4n}}}+\frac{R^{(2)}(z)}{|x|^{\frac{2n+1}{2n}}}
+O\left(|x|^{-\frac{3(2n+1)}{4n}}\right), \label{R-estimation}\\
\frac{\mathrm{d}}{\mathrm{d}z}R(z)
&=\frac{\frac{\mathrm{d}}{\mathrm{d}z}R^{(1)}(z)}{|x|^{\frac{2n+1}{4n}}}
+\frac{\frac{\mathrm{d}}{\mathrm{d}z}R^{(2)}(z)}{|x|^{\frac{2n+1}{2n}}}
+O\left(|x|^{-\frac{3(2n+1)}{4n}}\right), \label{dR-estimation}
\end{align}
uniformly for $z\in \mathbb{C} \setminus
(\partial\mathcal{D}(z_+)\cup\partial\mathcal{D}(z_-))
$.

We conclude this subsection with the calculations of the functions $R^{(1)}$ and  $R^{(2)}$ in \eqref{R-estimation} for later use. Inserting \eqref{Rjump-estimate} and \eqref{R-estimation} into the jump condition \eqref{eq:Rjump} for $R$ yields
\begin{equation}\label{eq:Rijumps}
R_{+}^{(1)}(z)=R_{-}^{(1)}(z)+J_R^{(1)}(z), \qquad R_{+}^{(2)}(z)=R_{-}^{(2)}(z)+R_{-}^{(1)}(z) J_R^{(1)}(z)+J_R^{(2)}(z),
\end{equation}
for $z\in\partial\mathcal{D}(z_+)\cup\partial\mathcal{D}(z_-)$. This, together with the facts that $R^{(1)}(z)=O\left(z^{-1}\right)$ and $R^{(2)}(z)=O\left(z^{-1}\right)$ as $z\to\infty$, implies that
\begin{equation}\label{R(1)-expression}
R^{(1)}(z)=\frac{1}{2\pi i}\int_{\partial\mathcal{D}(z_+)\cup\partial\mathcal{D}(z_-)}
\frac{J_R^{(1)}(\xi)}{\xi-z}\mathrm{d}\xi,
\end{equation}
and
\begin{equation}\label{R(2)-expression}
R^{(2)}(z)=\frac{1}{2 \pi i} \int_{\partial \mathcal{D}(z_+)\cup\partial\mathcal{D}(z_-)}\frac{R_{-}^{(1)}(\xi) J_R^{(1)}(\xi)+J_R^{(2)}(\xi)}{\xi-z}\mathrm{d}\xi.
\end{equation}
From the definition of $J_R^{(1)}$ given in \eqref{Rjump12}, we obtain from \eqref{zeta}, \eqref{R(1)-expression} and the residue theorem that
\begin{equation}\label{R(1)-expression-new}
R^{(1)}(z)=\frac{A^{(1)}}{z-z_+}-\frac{\sigma_1A^{(1)}\sigma_1}{z-z_-},\qquad z\in\mathbb{C}\setminus \left(\mathcal{D}(z_+)\cup\mathcal{D}(z_-)\right),
\end{equation}
where
\begin{align}\label{A(1)-expression}
A^{(1)}=\Res\limits_{z=z_+}J_R^{(1)}(z)=\begin{pmatrix} 0 & -\frac{\nu\rho{t}^{\nu}\zeta'(z_+)^{2\nu-1}}
{h_1\mathrm{e}^{2tg(z_+)}(2z_+)^{-2\nu}} \\ -\frac{h_1\mathrm{e}^{2tg(z_+)}(2z_+)^{-2\nu}}
{\rho{t}^{\nu}\zeta'(z_+)^{2\nu+1}} & 0 \end{pmatrix}.
\end{align}

To calculate $R^{(2)}$, it is readily seen from \eqref{Rjump12} and \eqref{R(2)-expression} that
\begin{equation}\label{R(2)-expression-new}
R^{(2)}(z)=\frac{A^{(2)}}{z-z_+}+\frac{B^{(2)}}{\left(z-z_+\right)^2}
-\frac{\sigma_1A^{(2)}\sigma_1}{z-z_-}+\frac{\sigma_1B^{(2)}\sigma_1}{\left(z-z_-\right)^2}, \qquad z\in\mathbb{C}\setminus \left(\mathcal{D}(z_+)\cup\mathcal{D}(z_-)\right),
\end{equation}
where
\begin{equation}\label{AB(2)-expression}
A^{(2)}=R^{(1)}_-\left(z_+\right) A^{(1)}+\Res\limits_{z=z_+}J_R^{(2)}(z),\qquad
B^{(2)}=\Res\limits_{z=z_+}\left((z-z_+) J_R^{(2)}(z)\right).
\end{equation}
From $J_R^{(2)}$ in \eqref{Rjump12}, it is easy to compute that
\begin{align}
&\Res\limits_{z=z_+}J_R^{(2)}(z)=-\frac{\zeta''(z_+)}{\zeta'(z_+)^3}\begin{pmatrix} \frac{\nu(\nu+1)}{2} & 0 \\ 0 & -\frac{\nu(\nu-1)}{2}\end{pmatrix},\label{residue-1}\\
&\Res\limits_{z=z_+}\left(\left(z-z_+\right) J_R^{(2)}(z)\right)=\frac{1}{\zeta'(z_+)^2}\begin{pmatrix} \frac{\nu(\nu+1)}{2} & 0 \\ 0 & -\frac{\nu(\nu-1)}{2}\end{pmatrix}.\label{residue-2}
\end{align}
It remains to compute $R^{(1)}_-\left(z_+\right)$. In view of \eqref{eq:Rijumps}, it follows from \eqref{R(1)-expression-new} that
\begin{align}\label{eq:R1z+}
&R^{(1)}_-\left(z_+\right)=-\frac{\sigma_1A^{(1)}\sigma_1}{z_+-z_-}
-\lim\limits_{z\to{z}_+}\left(J_R^{(1)}(z)-\frac{A^{(1)}}{z-z_+}\right)
=-\frac{\sigma_1A^{(1)}\sigma_1}{2z_+}\nonumber\\
&-\begin{pmatrix} 0 & -\frac{\nu\rho{t}^{\nu}\zeta'(z_+)^{2\nu-1}}
{h_1\mathrm{e}^{2tg(z_+)}(2z_+)^{-2\nu}}
\left(\frac{\nu}{z_+}+\frac{(2\nu-1)\zeta''(z_+)}{2\zeta'(z_+)}\right) \\ \frac{h_1\mathrm{e}^{2tg(z_+)}(2z_+)^{-2\nu}}
{\rho{t}^{\nu}\zeta'(z_+)^{2\nu+1}}
\left(\frac{\nu}{z_+}+\frac{(2\nu+1)\zeta''(z_+)}{2\zeta'(z_+)}\right) & 0 \end{pmatrix}.
\end{align}
Substituting \eqref{residue-1}--\eqref{eq:R1z+} into \eqref{AB(2)-expression}, we finally have that $A^{(2)}$ and $B^{(2)}$ in \eqref{R(2)-expression-new} are all diagonal matrices with
\begin{align}\label{B(2)-expression}
&B^{(2)}_{11}=\frac{\nu(\nu+1)}{2\zeta'(z_+)^2},\qquad
B^{(2)}_{22}=-\frac{\nu(\nu-1)}{2\zeta'(z_+)^2},\\
&A^{(2)}_{11}=-\frac{\nu^2}{z_+\zeta'(z_+)^2}-\frac{3\nu^2\zeta''(z_+)}{2\zeta'(z_+)^3}
-\frac{h_1^2\mathrm{e}^{4tg(z_+)}(2z_+)^{-4\nu-1}}{\rho^2t^{2\nu}\zeta'(z_+)^{4\nu+2}},
\label{A(2)-11-expression}\\
&A^{(2)}_{22}=\frac{\nu^2}{z_+\zeta'(z_+)^2}+\frac{3\nu^2\zeta''(z_+)}{2\zeta'(z_+)^3}
-\frac{\nu^2\rho^2\mathrm{e}^{-4tg(z_+)}(2z_+)^{4\nu-1}}{h_1^2t^{-2\nu}\zeta'(z_+)^{-4\nu+2}}.
\label{A(2)-22-expression}
\end{align}

\section{Asymptotic analysis of the RH problem for $\Psi$ with $\rho>1$}
\label{RHanalysis2}
In this section, we will perform  asymptotic analysis of the RH problem for $\Psi$ as $x\to -\infty$ with the parameter $\rho>1$.
As in the case of $0<\rho<1$, it consists of a series of explicit, invertible transformations, and the first few transformations $\Psi\to{X}\to{Y}\to{T}$ are exactly the same as  those defined in  \eqref{rescaling}, \eqref{Y-transformation} and \eqref{T-transformation}, respectively. The differences lie in the constructions of the global and the local parametrices. Indeed, by ignoring the exponentially small terms of $J_T$ shown in Figure \ref{GammaT}, we are led to the same RH problem for ${P}^{(\infty)}$ encountered in Section \ref{subsec:global1}. Since $\rho>1$, a solution is then given by \eqref{Pinfty} but with $\nu$ therein replaced by
\begin{equation}\label{nu0}
\nu
=-\frac{1}{2}-\frac{1}{2\pi{i}}\ln(\rho^2-1)
:=-\frac{1}{2}+\nu_0.
\end{equation}
It is clear that $\nu_0\in{i}\mathbb{R}$. If we continue to construct the local parametrix $P^{(+)}$ as in \eqref{P+} and \eqref{E-expression}, it comes out $P^{(+)}(z)P^{(\infty)}(z)^{-1}$ does not tend to $I$ as $t\to+\infty$ for $z\in \partial\mathcal{D}(z_+)$; see \eqref{P+Pinfty} with $\nu$ given by \eqref{nu0}. This means that the matching condition \eqref{matching} is not valid anymore. To overcome this difficulty, we need to make modifications of the global and the local parametrices.

\subsection{Global parametrix}
We modify the RH problem for $P^{(\infty)}$ by imposing specified singular behaviors at $z=z_{\pm}$, as stated below.
\subsection*{RH problem for $\widetilde{P}^{(\infty)}$}
\begin{description}
  \item{(1)} $\widetilde{P}^{(\infty)}(z)$ is analytic for $z\in \mathbb{C}\setminus [z_{-},z_{+}]$.
  \item{(2)} $\widetilde{P}^{(\infty)}(z)$ satisfies the jump condition
  \begin{equation}\label{Pinftyjump-m}
  \widetilde{P}^{(\infty)}_{+}(z)=\widetilde{P}^{(\infty)}_{-}(z)
  \begin{pmatrix}1-\rho^2 & 0 \\ 0 & (1-\rho^2)^{-1}\end{pmatrix},\qquad z\in(z_-,z_+).
  \end{equation}
  \item{(3)} As $z\rightarrow \infty$, we have $\widetilde{P}^{(\infty)}(z)=I+O(z^{-1})$.
  \item{(4)} $\widetilde{P}^{(\infty)}(z)$ behaves like $O((z-z_{\pm})^{-\frac{3}{2}})$ as $z \to z_{\pm}$.
\end{description}

We build a solution of this RH problem by defining
\begin{equation}\label{Pinftytilde}
\widetilde{P}^{(\infty)}(z)={V}(z)
\left(\frac{z-z_-}{z-z_+}\right)^{(\nu_0-\frac{1}{2})\sigma_3},
\arg(z-z_\pm)\in(-\pi,\pi),
\end{equation}
where
$\nu_0$ is a pure imaginary parameter given by \eqref{nu0} and
\begin{equation}\label{V-expression}
{V}(z)=I+\frac{A}{z-z_-}+\frac{B}{z-z_+}
\end{equation}
is a rational function in $z$ with the matrices $A=A(x)$ and $B=B(x)$ to be determined later.  The function $V$ also satisfies $\det{V}=1$ and the symmetry relation
 \begin{equation}\label{V-symmetry}
\sigma_1{V}(-z)\sigma_1=V(z),
\end{equation}
which follows from \eqref{T-symmetry}.   The asymptotic behaviors of $\widetilde{P}^{(\infty)}(z)$ as $z \to z_{\pm}$ given in condition (4) of the above RH problem then follows from \eqref{Pinftytilde} and \eqref{V-expression}. 

We introduce the function $V$ in \eqref{Pinftytilde} in order to fulfill the matching conditions \eqref{matching-m} below. Similar techniques have been applied in \cite{XXZ2021} to study singular asymptotics of  the Painlev\'{e} IV transcendents, and in \cite{ZXZ2011,ZZ2008} to derive uniform asymptotic approximations of some orthogonal polynomials.

\subsection{Local parametrices at $z_{\pm}$}
In the neighborhoods $\mathcal{D}(z_{\pm})$ of two saddle points $z_{\pm}$, we look for two local parametrices $\widetilde{P}^{(\pm)}$ that satisfy items (1) and (2) of the RH problems for $P^{(\pm)}$, and the  matching conditions
\begin{equation}\label{matching-m}
\widetilde{P}^{(\pm)}(z)=\left(I+O(t^{-1})\right)\widetilde{P}^{(\infty)}(z),   \qquad  t=|x|^{\frac{2n+1}{2n}}\to+\infty,
\end{equation}
for $z \in \partial\mathcal{D}(z_{\pm})$, where $\widetilde{P}^{(\infty)}$ is given in \eqref{Pinftytilde}.
%
%

Similar to the construction of ${P}^{(+)}$ in \eqref{P+}, the parametrix $\widetilde{P}^{(+)}$ is defined by
\begin{equation}\label{P+tilde}
\widetilde{P}^{(+)}(z)=\widetilde{E}(z)\Phi_{(\mathrm{PC})}\left(\sqrt{t}\zeta(z)\right)
\left(\frac{h_1}{\rho}\right)^{\frac{\sigma_3}{2}}\sigma_3\mathrm{e}^{tg(z)\sigma_3},
\end{equation}
where $\Phi_{(\mathrm{PC})}$ is the parabolic cylinder parametrix with parameter $\nu$ given by \eqref{nu0},
$\zeta(z)$ is the conformal mapping \eqref{zeta}, $h_1=-i\sqrt{2\pi}\mathrm{e}^{i\pi\nu_0}/\Gamma(\frac{1}{2}-\nu_0)$ and
\begin{equation}\label{Etilde-expression}
\widetilde{E}(z)=\widetilde{P}^{(\infty)}(z)\sigma_3
\left(\frac{\rho}{h_1}\right)^{\frac{\sigma_3}{2}}
\left(\sqrt{t}\zeta(z)\right)^{\nu\sigma_3}\mathrm{e}^{-tg(z_+)\sigma_3}
\begin{pmatrix}
1 & 0 \\ -\frac{1}{\sqrt{t}\zeta(z)} & 1
\end{pmatrix}.
\end{equation}

By \eqref{Pinftyjump-m}, one can check directly that $\widetilde{E}$ is analytic in the punctured disc $\mathcal{D}(z_{+})\setminus\{z_+\}$.
To make sure that $z_+$ is a removable singular point of $\widetilde{E}$, we need to choose appropriate matrices $A$ and $B$ in \eqref{V-expression}. Using \eqref{Pinfty}, \eqref{zeta}, \eqref{V-expression} and \eqref{Etilde-expression}, it is readily seen that
\begin{equation}\label{Etilde-expansion}
\widetilde{E}(z)=\left(\frac{B}{z-z_+}+I+\frac{A}{2z_+}+O(z-z_+)\right)\left[I+
\begin{pmatrix} 0 & 0 \\ \frac{2z_+c}{z-z_+} & 0\end{pmatrix}\right]K(z),\quad z \to z_+,
\end{equation}
where $K(z)$ is analytic at $z_+$ and
\begin{align}\label{c}
c&=c(x)=\frac{h_1\mathrm{e}^{2g(z_+)t}}
{\rho(2z_+)^{2\nu_0}\zeta'(z_+)^{2\nu_0}t^{\nu_0}}\nonumber\\
&=\exp\left\{2g(z_+)t-\nu_0\ln\left[4iz_+^2\zeta'(z_+)^2t\right]
-i\arg\Gamma\left(\frac{1}{2}-\nu_0\right)-\frac{i\pi}{2}\right\}
\end{align}
with $\nu_0\in i\mathbb{R}$ given in \eqref{nu0}. On one hand, the analyticity of $\widetilde{E}$ at $z_+$ implies that the matrices ${A}$ and ${B}$ should satisfy
\begin{equation}\label{ABrelation1}
\left(I+\frac{A}{2z_+}\right)\begin{pmatrix} 0 & 0 \\ 2z_+c & 0\end{pmatrix}+B=0.
\end{equation}
On the other hand, the symmetry relation \eqref{V-symmetry} gives us
\begin{equation}\label{ABrelation2}
\sigma_1A\sigma_1=-B.
\end{equation}
Note that the equation \eqref{ABrelation1} also shows that $B \begin{pmatrix} 0 & 0 \\ \frac{2z_+c}{z-z_+} & 0\end{pmatrix}=0$. One can solve the equations \eqref{ABrelation1} and \eqref{ABrelation2} explicitly, and the solutions are given by
\begin{equation}\label{ABexpression}
{A}=\begin{pmatrix} 0 & \frac{2z_+c}{1-c^2}\\[.1cm]     0 & \frac{2z_+c^2}{1-c^2} \end{pmatrix},\qquad
{B}=\begin{pmatrix} -\frac{2z_+c^2}{1-c^2} & 0\\[.1cm]  -\frac{2z_+c}{1-c^2} & 0 \end{pmatrix}.
\end{equation}
This in turn shows that the determinant of $V$ in \eqref{V-expression} is indeed equal to 1. Moreover, it is readily seen from \eqref{Pinftytilde}, \eqref{V-expression}, \eqref{ABexpression} and \eqref{Pinfty-expansion} that
\begin{equation}\label{Pinftyatinfty}
\widetilde{P}^{(\infty)}(z)=I+\frac{\widetilde{P}_1^{(\infty)}}{z}
+O(z^{-2}), \qquad z\to \infty,
\end{equation}
where
\begin{equation}\label{eq:tildeP1infty}
\widetilde{P}_1^{(\infty)}=\begin{pmatrix}
2z_+\nu-\frac{2z_+c^2}{1-c^2} & \frac{2z_+c}{1-c^2}\\
-\frac{2z_+c}{1-c^2} & -2z_+\nu+\frac{2z_+c^2}{1-c^2}
\end{pmatrix}
\end{equation}
with $\nu$ and $c$ being given in \eqref{nu0} and \eqref{c}, respectively.

Since now $\widetilde{E}$ is analytic in $\mathcal{D}(z_+)$, it then follows from straightforward calculations that, as what we have done in Section \ref{subsec:localpara}, $\widetilde{P}^{(+)}$ defined in \eqref{P+tilde} is analytic in $\mathcal{D}(z_{+}) \setminus (\cup_{k=1,4, 5, 6,7} \Sigma_{k})$, satisfies the jump condition \eqref{eq:P+jump} and the matching condition \eqref{matching-m}. In view of the symmetry relation \eqref{T-symmetry}, we define, as in \eqref{P-},
the local parametrix $\widetilde{P}^{(-)}$ by
\begin{equation}\label{P-tilde}
\widetilde{P}^{(-)}(z)=\sigma_1\widetilde{P}^{(+)}(-z)\sigma_1,\qquad z\in\mathcal{D}(z_{-}).
\end{equation}

%


We conclude this subsection by noting  that \eqref{ABexpression} only holds for all sufficiently large $|x|$ lying outside the zero set of the function $1-c^2(x)$,
which consists of a sequence of points $\{x_m\}_{m\in \mathbb{N}}$ on the negative real axis determined by the equation
\begin{multline}\label{poles-equation}
-2ig(z_+)|x_m|^{\frac{2n+1}{2n}}
+\frac{2n+1}{2n}i\nu_0\ln|x_m|+i\nu_0\ln\left[4iz_+^2\zeta'(z_+)^2\right]
\\
-\arg\Gamma\left(\frac{1}{2}-\nu_0\right)-\frac{\pi}{2}=-m\pi.
\end{multline}
In view of \eqref{g-expression}, \eqref{zeta'behavior} and \eqref{nu0},
the  quantities $ig(z_+)$,  $i\zeta'(z_+)^2$ and $i\nu_0$ in the above equation are all real.
As we will see later, the functions $q_{n}((-1)^{n+1}x;\rho)$ possess a sequence of simple poles $\{p_m\}$, such that $p_m\sim x_m$ as $m\to\infty$; see \eqref{poles} in Corollary \ref{cor}.


\subsection{Small norm RH problem}
As in \eqref{Final-transformation}, the final transformation is defined as
\begin{equation}\label{Final-transformation-m}
\widetilde{R}(z)=\left\{\begin{aligned}
&{T}(z)\widetilde{P}^{(\pm)}(z)^{-1},\quad &z&\in \mathcal{D}(z_{\pm})\setminus(\cup_1^7\Sigma_k),\\
&{T}(z)\widetilde{P}^{(\infty)}(z)^{-1},\quad  &\mathrm{e}&\mathrm{lsewhere}. \end{aligned}
\right.
\end{equation}
It is then readily seen that $\widetilde{R}$ solves the following RH problem.

\subsection*{RH problem for $\widetilde{R}$}
\begin{description}
\item{(1)} $\widetilde{R}(z)$ is analytic for $z\in \mathbb{C}\setminus\Sigma_{R}$, where the contours $\Sigma_{R}$ is defined in \eqref{def:SigmaR}; see also Figure \ref{GammaR} for an illustration.
\item{(2)} $\widetilde{R}(z)$ satisfies the jump condition $$\widetilde{R}_+(z)=\widetilde{R}_-(z)J_{\widetilde{R}}(z),$$
  where
 \begin{equation}\label{Jump-tilde-R}
 J_{\widetilde{R}}(z)=\left\{\begin{aligned}
&\widetilde{P}^{(\pm)}(z)\widetilde{P}^{(\infty)}(z)^{-1},\quad &z&\in \partial\mathcal{D}(z_{\pm}),\\
&\widetilde{P}^{(\infty)}(z)J_{T}(z)\widetilde{P}^{(\infty)}(z)^{-1},\quad &z& \in \Sigma_R\setminus (\partial \mathcal{D}(z_+) \cup \partial \mathcal{D}(z_-)).\\ \end{aligned}
\right.
 \end{equation}
\item{(3)} As $z\rightarrow\infty$, we have
\begin{equation}\label{Rexpan}
\widetilde{R}(z)=I+\frac{\widetilde{R}_1}{z}+O\left(z^{-2}\right),
\end{equation}
where $\widetilde R_1$ is independent of $z$.
\end{description}

In view of the matching conditions \eqref{matching-m}, we have that, as discussed in Section \ref{subsec:final},
\begin{equation}\label{JRestimation-m}
J_{\widetilde{R}}(z)=\left\{\begin{aligned}
&I+O\left(t^{-1}\right), \quad &z&\in\partial \mathcal{D}(z_{\pm}),\\
&I+O\left(\mathrm{e}^{-c_2t}\right), \quad &z& \in \Sigma_R\setminus (\partial \mathcal{D}(z_+) \cup \partial \mathcal{D}(z_-)),
\end{aligned}\right.
\end{equation}
as $t=|x|^{\frac{2n+1}{2n}}\to +\infty$, where $c_2>0$ is some constant. As a consequence, it follows that
\begin{equation}\label{Rtilde-estimation}
\widetilde{R}(z)=I+O\left(t^{-1}\right),\qquad  t \to +\infty,
\end{equation}
uniformly for $z\in \mathbb{C}\setminus\Sigma_{R}$.


\section{Proofs of main results}
\label{proof-largegap}
In this section, we prove our main results, namely, Theorems \ref{q-asymptotic}, \ref{thm:largegap},\ref{thm1} and Corollary \ref{cor}. The proofs rely on the identities \eqref{q-expression}, \eqref{Fred-expression-to-x} and \eqref{Fred-expression-to-rho}, and are outcomes of asymptotic analysis of the RH problem for $\Psi$ given in the previous two sections. Below, it is understood that $0<\rho<1$ in Sections \ref{subsec:proof1} and \ref{subsec:proof2}, and $\rho>1$ in Sections \ref{proof1} and \ref{proof:cor}.

\subsection{Proof of Theorem \ref{q-asymptotic} and partial proof of Theorem \ref{thm:largegap}}\label{subsec:proof1}
By inverting the transformations $\Psi\to{X}\to{Y}\to{T}\to{R}$ given in \eqref{rescaling},  \eqref{Y-transformation}, \eqref{T-transformation} and \eqref{Final-transformation}, it follows that
\begin{equation}
\Psi\left(|x|^{\frac{1}{2n}}z\right)
=R(z)P^{(\infty)}(z),\qquad z\in\mathbb{C}\setminus (\cup_{k=1}^7\Sigma_i \cup  \mathcal{D}(z_+)\cup  \mathcal{D}(z_-)).
\end{equation}
Let $z\to \infty$ in the above formula, we obtain from \eqref{q-expression}, \eqref{Fred-expression-to-x} and the large-$z$ asymptotics of $P^{(\infty)}$ and $R$ given in \eqref{Pinfty-expansion} and \eqref{R-expansion} that
\begin{equation}\label{q0}
q_n\left((-1)^{n+1}x;\rho\right)=2i|x|^{\frac{1}{2n}}(R_1)_{12},
\end{equation}
and
\begin{equation}\label{F_x}
\partial_x\ln{F}_n(x;\rho)=2i|x|^{\frac{1}{2n}}\left(2z_+\nu+(R_1)_{11}\right),
\end{equation}
where $R_1$ is shown in \eqref{R-expansion}. This, together with \eqref{R-estimation}, \eqref{R(1)-expression-new} and \eqref{R(2)-expression-new}, implies that, as $x\to-\infty$,
\begin{equation}\label{q0-1}
q_n\left((-1)^{n+1}x;\rho\right)=2i|x|^{-\frac{2n-1}{4n}}
\left(A^{(1)}_{12}-A^{(1)}_{21}\right)+O\left(|x|^{-1}\right),
\end{equation}
and
\begin{equation}\label{F_x-1}
\partial_x\ln{F}_n(x;\rho)=2i|x|^{\frac{1}{2n}}
\left[2z_+\nu+\left(A^{(2)}_{11}-A^{(2)}_{22}\right)|x|^{-\frac{2n+1}{2n}}
+O\left(|x|^{-\frac{3(2n+1)}{4n}}\right)\right].
\end{equation}
With the aid of the explicit formulas of $A^{(1)}$, $A^{(2)}_{11}$ and $A^{(2)}_{22}$ given in \eqref{A(1)-expression}, \eqref{A(2)-11-expression} and \eqref{A(2)-22-expression}, it follows that
\begin{align}\label{q0-2}
q_n\left((-1)^{n+1}x;\rho\right)&=\frac{2i}{|x|^{\frac{2n-1}{4n}}}\Bigg(
\frac{h_1(2z_+)^{-2\nu}\mathrm{e}^{2g(z_+)t}}
{\rho{t}^{\nu}\zeta'(z_+)^{2\nu+1}}-\frac{\nu\rho(2z_+)^{2\nu}\zeta'(z_+)^{2\nu-1}}
{h_1t^{-\nu}\mathrm{e}^{2g(z_+)t}}
\Bigg)+O\left(|x|^{-1}\right),
\end{align}
and
\begin{align}\label{F_x-2}
\partial_x\ln{F}_n(x;\rho)=2i|x|^{\frac{1}{2n}}
\Bigg[&2z_+\nu+\Bigg(
\frac{\nu^2\rho^2\mathrm{e}^{-4g(z_+)t}(2z_+)^{4\nu-1}}{h_1^2t^{-2\nu}\zeta'(z_+)^{-4\nu+2}}
-\frac{h_1^2\mathrm{e}^{4g(z_+)t}(2z_+)^{-4\nu-1}}{\rho^2t^{2\nu}\zeta'(z_+)^{4\nu+2}}
\nonumber\\
&\qquad-\frac{2\nu^2}{z_+\zeta'(z_+)^2}-\frac{3\nu^2\zeta''(z_+)}{\zeta'(z_+)^3}
\Bigg)|x|^{-\frac{2n+1}{2n}}+O\left(|x|^{-\frac{3(2n+1)}{4n}}\right)\Bigg],
\end{align}
as $x\to-\infty$.

We next insert $\zeta'(z_+)$ and $\zeta''(z_+)$ in \eqref{zeta'behavior} and \eqref{zeta''behavior} into \eqref{q0-2} and \eqref{F_x-2}, and obtain that,
\begin{align}\label{q0-3}
q_n\left((-1)^{n+1}x;\rho\right)=\frac{|x|^{-\frac{2n-1}{4n}}}{\mathrm{e}^{\frac{\pi i}{4}}\sqrt{2n}}\Bigg(\frac{h_1\mathrm{e}^{2g(z_+)t}}
{\rho\,\mathrm{e}^{\frac{3\pi i\nu}{2}}(8nt)^{\nu}}
-\frac{\nu\rho\,\mathrm{e}^{\frac{3\pi i\nu}{2}}(8nt)^{\nu}}
{h_1\mathrm{e}^{2g(z_+)t}}
\Bigg)&+O\left(|x|^{-\frac{2n+3}{4n}}\right),
\end{align}
and
\begin{multline}\label{F_x-3}
\partial_x\ln{F}_n(x;\rho)=4i\nu z_+|x|^{\frac{1}{2n}}+\frac{2n+1}{2n}\nu^2|x|^{-1}
-\frac{1}{4n}\Bigg(
\frac{\nu^2\rho^2(8nt)^{2\nu}\mathrm{e}^{3\pi i\nu}}{h_1^2\mathrm{e}^{4g(z_+)t}}
-\frac{h_1^2\mathrm{e}^{4g(z_+)t}}{\rho^2\mathrm{e}^{3\pi i\nu}(8nt)^{2\nu}}\Bigg)|x|^{-1}
\\
+O\left(|x|^{-1-\frac{1}{n}}\right).
\end{multline}
To proceed, note that $\nu=-\frac{1}{2\pi i}\ln(1-\rho^2)\in{i}\mathbb{R}$, it then follows from the reflection formula of Gamma function that
\begin{equation}\label{Gamma(nu)}
|\Gamma(\nu)|^2=\Gamma(\nu)\Gamma(-\nu)=-\frac{\pi}{\nu\sin(\nu\pi)}
=\frac{2\pi}{i\nu\rho^2}\mathrm{e}^{-i\pi\nu}.
\end{equation}
By further substituting $h_1=\sqrt{2\pi}\mathrm{e}^{i\pi\nu}/\Gamma(-\nu)$ and the expansion of $z_+$ shown in \eqref{z+-expansion} into \eqref{q0-3} and \eqref{F_x-3}, it is readily seen from \eqref{Gamma(nu)} and the fact $t=|x|^{\frac{2n+1}{2n}}$ that,
\begin{equation}\label{q0-4}
q_n\left((-1)^{n+1}x;\rho\right)=\frac{\sqrt{-\ln(1-\rho^2)}}{\sqrt{n\pi}|x|^{\frac{2n-1}{4n}}}
\cos\left(\psi(x,\rho)+\frac{\pi}{4}\right)
+O\left(|x|^{-\frac{2n+3}{4n}}\right),
\end{equation}
and
\begin{align}\label{eq: Fx-4}
\partial_x\ln{F}_n(x;\rho)
=2i\nu\sum_{k=0}^{n}a_k|x|^{\frac{1-2k}{2n}}-\frac{2n+1}{2nx}\nu^2
-\frac{i\nu}{2nx}\cos\left(2\psi(x,\rho)\right)+O\left(|x|^{-1-\frac{1}{2n}}\right),
\end{align}
where $a_k$, $k=0,1,\ldots,n$, are given by \eqref{ak-expression} and
\begin{equation}\label{psi}
\psi(x,\rho)=2ig(z_+)|x|^{\frac{2n+1}{2n}}-i\nu\ln\left(8n|x|^{\frac{2n+1}{2n}}\right)
+\arg\Gamma(-\nu).
\end{equation}

Finally, by noting that $\beta=2\nu i$ (see \eqref{beta} and \eqref{nu}), a combination of \eqref{eq:gexpansion}, \eqref{q0-4} and  \eqref{psi} gives us \eqref{qAymp-1} and \eqref{beta}. This completes the proof of Theorem \ref{q-asymptotic}. Integrating both sides of \eqref{eq: Fx-4}, it follows that
\begin{align}\label{Fn-asymp}
\ln{F}_n(x;\rho)=
-2i\nu\sum_{k=0}^{n}\frac{2na_k}{1+2(n-k)}&|x|^{\frac{1+2(n-k)}{2n}}
-\frac{2n+1}{2n}\nu^2\ln|x|+\mathcal{C}+O\left(|x|^{-\frac{1}{2n}}\right),
\end{align}
where $\mathcal{C}$ is a constant of integration that is independent of $x$. We evaluate this constant in the next section, which will lead to a completed proof of Theorem \ref{thm:largegap}.


\subsection{Proof of Theorem \ref{thm:largegap}: completion}\label{subsec:proof2}
To evaluate the constant $\mathcal{C}$ in \eqref{Fn-asymp}, we work on the differential identity of $\ln{F}_n$ with respect to $\rho$ given in \eqref{Fred-expression-to-rho}; see also \cite{CC2021} for a similar strategy. Since the arguments are rather involved, they are split into a series of lemmas.

We start with the following decomposition of $\partial_{\rho}\ln F_n(x;\rho)$ as $t=|x|^{\frac {2n+1}{2n}}\to+\infty$.
\begin{lemma}\label{Fred-estimate}
As $t\to+\infty$, we have
\begin{equation}\label{eq:dFI}
\partial_{\rho}\ln F_n(x;\rho)=I_1+I_2+2I_{z_+}+O\left(\mathrm{e}^{-c_3t}\right),
\end{equation}
where $c_3>0$ is a constant, and
\begin{align}\label{I-1}
I_1&=\frac{4t\rho}{1-\rho^2}\int_{[z_-,z_+]}g'(z)\frac{\mathrm{d}z}{2\pi i}
=\frac{4\rho{g}(z_+)}{(1-\rho^2)\pi{i}}t,\\
I_2&=-\frac{4\rho}{1-\rho^2}\int_{[0,z_+]}\mathrm{Tr}
\left[T^{-1}(z)T'(z)\sigma_3\right]\frac{\mathrm{d}z}{2\pi i},\label{I-2}\\
I_{z_+}&=I_{z_+,1}+I_{z_+,2}+I_{z_+,3}+I_{z_+,4}, \label{I-z+}
\end{align}
with
\begin{align}
I_{z_+,1}&=\int_{\Sigma_1\cap \mathcal{D}(z_+)}\mathrm{e}^{2tg(z)}
\mathrm{Tr}\left[T^{-1}(z)T'(z)\sigma_-\right]\frac{\mathrm{d}z}{2\pi i}, \label{eq:z+1}\\
I_{z_+,2}&= -\int_{\Sigma_4\cap \mathcal{D}(z_+)}\mathrm{e}^{-2tg(z)}
\mathrm{Tr}\left[T^{-1}(z)T'(z)\sigma_+\right]\frac{\mathrm{d}z}{2\pi i},\label{eq:z+2}\\
I_{z_+,3}&=\frac{1+\rho^2}{(1-\rho^2)^2}
\int_{\Sigma_7\cap \mathcal{D}(z_+)}\mathrm{e}^{2tg(z)}
\mathrm{Tr}\left[T^{-1}(z)T'(z)\sigma_-\right]\frac{\mathrm{d}z}{2\pi i},\label{eq:z+3}\\
I_{z_+,4}&=-\frac{1+\rho^2}{(1-\rho^2)^2}
\int_{\Sigma_6\cap \mathcal{D}(z_+)}\mathrm{e}^{-2tg(z)}
\mathrm{Tr}\left[T^{-1}(z)T'(z)\sigma_+\right]\frac{\mathrm{d}z}{2\pi i}.\label{eq:z+4}
\end{align}
\end{lemma}

\begin{proof}
Using the differential identity \eqref{Fred-expression-to-rho} and the transformation $\Psi\to{X}$ defined in \eqref{rescaling}, it follows that
\begin{equation}\label{eq:dF}
\partial_{\rho}\ln{F_n(x;\rho)}=\frac{1}{\rho}\int_{\Gamma_+\cup\Gamma_-}
\mathrm{Tr}\left[\Psi^{-1}(z)\Psi'(z)(J_{\Psi}(z)-I)\right]\frac{\mathrm{d}z}{2\pi i}:=I_{\gamma_+}+I_{\gamma_-},
\end{equation}
where
\begin{align}\label{I-gamma+0}
I_{\gamma_+}&=\int_{\gamma_+}\mathrm{e}^{2tg(z)}
\mathrm{Tr}\left[X^{-1}(z)X'(z)\sigma_-\right]\frac{\mathrm{d}z}{2\pi i},\\
I_{\gamma_-}&=\int_{\gamma_-}\mathrm{e}^{-2tg(z)}
\mathrm{Tr}\left[X^{-1}(z)X'(z)\sigma_+\right]\frac{\mathrm{d}z}{2\pi i},\label{I-gamma-0}
\end{align}
with
$$
\sigma_-=\begin{pmatrix} 0 & 0 \\ 1 & 0 \end{pmatrix}\quad \mathrm{and} \quad
\sigma_+=\begin{pmatrix} 0 & 1 \\ 0 & 0 \end{pmatrix}.
$$
According to the transformation $X \to Y$ in \eqref{Y-transformation}, we could deform the contours of integration in \eqref{I-gamma+0} and \eqref{I-gamma-0} to the rays $\Sigma_i$, $i=1,\ldots,5$, shown in Figure \ref{GammaS}, and obtain
\begin{align}\label{I-gamma+}
I_{\gamma_+}&=\int_{\Sigma_1\cup\Sigma_2}\mathrm{e}^{2tg(z)}
\mathrm{Tr}\left[Y^{-1}(z)Y'(z)\sigma_-\right]\frac{\mathrm{d}z}{2\pi i}+\int_{\Sigma_5}\mathrm{e}^{2tg(z)}
\mathrm{Tr}\left[Y_+^{-1}(z)Y'_+(z)\sigma_-\right]\frac{\mathrm{d}z}{2\pi i},\\
-I_{\gamma_-}&=\int_{\Sigma_3\cup\Sigma_4}\mathrm{e}^{-2tg(z)}
\mathrm{Tr}\left[Y^{-1}(z)Y'(z)\sigma_+\right]\frac{\mathrm{d}z}{2\pi i}+\int_{\Sigma_5}\mathrm{e}^{-2tg(z)}
\mathrm{Tr}\left[Y_-^{-1}(z)Y'_-(z)\sigma_+\right]\frac{\mathrm{d}z}{2\pi i}.\label{I-gamma-}
\end{align}

If $z\in\Sigma_5$, it follows from the transformation $Y \to T$ given in \eqref{T-transformation} and \eqref{H-expression} that
\begin{align*}
&\mathrm{Tr}\left[Y^{-1}_{\pm}(z)Y'_{\pm}(z)\sigma_{\mp}\right]=
\mathrm{Tr}\left[T^{-1}_{\pm}(z)T'_{\pm}(z)H^{-1}_{\pm}(z)\sigma_{\mp}H_{\pm}(z)\right]+
\mathrm{Tr}\left[H_{\pm}(z)(H_{\pm}^{-1})'(z)\sigma_{\mp}\right],\\
&\mathrm{e}^{\pm2tg(z)}\mathrm{Tr}\left[H_{\pm}(z)(H_{\pm}^{-1})'(z)\sigma_{\mp}\right]=
\pm\frac{2t\rho}{1-\rho^2}g'(z),\\
&\mathrm{e}^{2tg(z)}H_{+}^{-1}(z)\sigma_{-}H_{+}(z)=\begin{pmatrix} -\frac{\rho}{1-\rho^2} & -\frac{\rho^2}{(1-\rho^2)^2}\mathrm{e}^{-2tg(z)} \\
\mathrm{e}^{2tg(z)} & \frac{\rho}{1-\rho^2} \end{pmatrix},\\
&\mathrm{e}^{-2tg(z)}H_{-}^{-1}(z)\sigma_{+}H_{-}(z)=\begin{pmatrix} \frac{\rho}{1-\rho^2} & \mathrm{e}^{-2tg(z)} \\ -\frac{\rho^2}{(1-\rho^2)^2}\mathrm{e}^{2tg(z)}
 & -\frac{\rho}{1-\rho^2} \end{pmatrix},
\end{align*}
where $H$ is defined in \eqref{H-expression}. Thus, we have
\begin{align}\nonumber
&\int_{\Sigma_5}\mathrm{e}^{2tg(z)}
\mathrm{Tr}\left[Y_+^{-1}(z)Y'_+(z)\sigma_-\right]\frac{\mathrm{d}z}{2\pi i}\\ \nonumber
&=\frac{2t\rho}{1-\rho^2}\int_{\Sigma_5}g'(z)\frac{\mathrm{d}z}{2\pi i}-\frac{\rho}{1-\rho^2}\int_{\Sigma_5}\mathrm{Tr}\left[T^{-1}(z)T'(z)\sigma_3\right]
\frac{\mathrm{d}z}{2\pi i}\\ \label{eq:intSigma5a}
&~~~+\frac{1}{(1-\rho^2)^2}\int_{\Sigma_7}\mathrm{e}^{2tg(z)}
\mathrm{Tr}\left[T^{-1}(z)T'(z)\sigma_-\right]\frac{\mathrm{d}z}{2\pi i}\nonumber\\
&~~~-\frac{\rho^2}{(1-\rho^2)^2}\int_{\Sigma_6}\mathrm{e}^{-2tg(z)}
\mathrm{Tr}\left[T^{-1}(z)T'(z)\sigma_+\right]\frac{\mathrm{d}z}{2\pi i}
\end{align}
and
\begin{align} \nonumber
&\int_{\Sigma_5}\mathrm{e}^{-2tg(z)}
\mathrm{Tr}\left[Y_-^{-1}(z)Y'_-(z)\sigma_+\right]\frac{\mathrm{d}z}{2\pi i}\\  \nonumber
&=\frac{-2t\rho}{1-\rho^2}\int_{\Sigma_5}g'(z)\frac{\mathrm{d}z}{2\pi i}+\frac{\rho}{1-\rho^2}\int_{\Sigma_5}\mathrm{Tr}\left[T^{-1}(z)T'(z)\sigma_3\right]
\frac{\mathrm{d}z}{2\pi i} \\ \label{eq:intSigma5b}
&~~~+\frac{1}{(1-\rho^2)^2}\int_{\Sigma_6}\mathrm{e}^{-2tg(z)}
\mathrm{Tr}\left[T^{-1}(z)T'(z)\sigma_+\right]\frac{\mathrm{d}z}{2\pi i}\nonumber\\
&~~~-\frac{\rho^2}{(1-\rho^2)^2}\int_{\Sigma_7}\mathrm{e}^{2tg(z)}
\mathrm{Tr}\left[T^{-1}(z)T'(z)\sigma_-\right]\frac{\mathrm{d}z}{2\pi i}.
\end{align}
It is worth to mention that we have used the relations
\begin{align*}\int_{\Sigma_5}\mathrm{Tr}\left[T^{-1}_{+}(z)T'_{+}(z)\sigma_{-}\right]\mathrm{d}z
&=\frac{1}{(1-\rho^2)^2}\int_{\Sigma_{5}}\mathrm{Tr}\left[T^{-1}_{-}(z)T'_{-}(z)\sigma_{-}\right]\mathrm{d}z\\
&=\frac{1}{(1-\rho^2)^2}\int_{\Sigma_{7}}\mathrm{Tr}\left[T^{-1}(z)T'(z)\sigma_{-}\right]\mathrm{d}z
\end{align*}
and
\begin{align*}\int_{\Sigma_5}\mathrm{Tr}\left[T^{-1}_{-}(z)T'_{-}(z)\sigma_{+}\right]\mathrm{d}z
&=\frac{1}{(1-\rho^2)^2}\int_{\Sigma_{5}}\mathrm{Tr}\left[T^{-1}_{+}(z)T'_{+}(z)\sigma_{+}\right]\mathrm{d}z\\
&=\frac{1}{(1-\rho^2)^2}\int_{\Sigma_{6}}\mathrm{Tr}\left[T^{-1}(z)T'(z)\sigma_{+}\right]\mathrm{d}z
\end{align*}in \eqref{eq:intSigma5a} and \eqref{eq:intSigma5b}, respectively. Moreover,
we observe from the cyclic property of the trace and the structure of the jump matrix $J_T$ for $T$ that the integrals involving $T$ in  \eqref{eq:intSigma5a} and \eqref{eq:intSigma5b} have the same values when integrating along the positive or the negative side of the contours $\Sigma_k$, $k=5,6,7$. This justifies our
removal of the subscripts in $T$.


Note that $T(z)=Y(z)$ for $z\in \Sigma_k$, $k=1,\ldots,4$ (see \eqref{T-transformation} and \eqref{H-expression}),
by substituting \eqref{eq:intSigma5a} and \eqref{eq:intSigma5b} into \eqref{I-gamma+} and \eqref{I-gamma-}, we have
\begin{align*}
I_{\gamma_+}&=\int_{\Sigma_1\cup\Sigma_2}\mathrm{e}^{2tg(z)}
\mathrm{Tr}\left[T^{-1}(z)T'(z)\sigma_-\right]\frac{\mathrm{d}z}{2\pi i}\\
&~~~+\frac{2t\rho}{1-\rho^2}\int_{\Sigma_5}g'(z)\frac{\mathrm{d}z}{2\pi i}-\frac{\rho}{1-\rho^2}\int_{\Sigma_5}\mathrm{Tr}\left[T^{-1}(z)T'(z)\sigma_3\right]
\frac{\mathrm{d}z}{2\pi i}\\
&~~~+\frac{1}{(1-\rho^2)^2}\int_{\Sigma_7}\mathrm{e}^{2tg(z)}
\mathrm{Tr}\left[T^{-1}(z)T'(z)\sigma_-\right]\frac{\mathrm{d}z}{2\pi i}\\
&~~~-\frac{\rho^2}{(1-\rho^2)^2}\int_{\Sigma_6}\mathrm{e}^{-2tg(z)}
\mathrm{Tr}\left[T^{-1}(z)T'(z)\sigma_+\right]\frac{\mathrm{d}z}{2\pi i}
\end{align*}
and
\begin{align*}
I_{\gamma_-}&=-\int_{\Sigma_3\cup\Sigma_4}\mathrm{e}^{-2tg(z)}
\mathrm{Tr}\left[T^{-1}(z)T'(z)\sigma_+\right]\frac{\mathrm{d}z}{2\pi i}\\
&~~~+\frac{2t\rho}{1-\rho^2}\int_{\Sigma_5}g'(z)\frac{\mathrm{d}z}{2\pi i}-\frac{\rho}{1-\rho^2}\int_{\Sigma_5}\mathrm{Tr}\left[T^{-1}(z)T'(z)\sigma_3\right]
\frac{\mathrm{d}z}{2\pi i}\\
&~~~-\frac{1}{(1-\rho^2)^2}\int_{\Sigma_6}\mathrm{e}^{-2tg(z)}
\mathrm{Tr}\left[T^{-1}(z)T'(z)\sigma_+\right]\frac{\mathrm{d}z}{2\pi i}\\
&~~~+\frac{\rho^2}{(1-\rho^2)^2}\int_{\Sigma_7}\mathrm{e}^{2tg(z)}
\mathrm{Tr}\left[T^{-1}(z)T'(z)\sigma_-\right]\frac{\mathrm{d}z}{2\pi i}.
\end{align*}
This, together with the symmetry relation  \eqref{T-symmetry} and the fact that
$$
\Re g(z)\left\{
          \begin{array}{ll}
            >\varrho, & \hbox{$z\in \cup_{k=1,2,7}\Sigma_k\setminus \mathcal D(z_+),$} \\
            <\varrho, & \hbox{$z \in \cup_{k=3,4,6}\Sigma_k\setminus \mathcal D(z_-),$}
          \end{array}
        \right.
$$
for some $\varrho>0$, implies that
\begin{equation}\label{eq:Iest}
I_{\gamma_+}+I_{\gamma_-}=I_1+I_2+2I_{z_+}+O\left(\mathrm{e}^{-c_3t}\right) ~~\mathrm{as}~~t\to+\infty,
\end{equation}
where $I_1$, $I_2$ and $I_{z_+}$ are defined by \eqref{I-1}, \eqref{I-2} and \eqref{I-z+}--\eqref{eq:z+4}, respectively.
The equation \eqref{eq:dFI} then follows directly from \eqref{eq:dF} and \eqref{eq:Iest}.
\end{proof}

We next estimate the integrals $I_{z_+}$ and $I_2$ in  \eqref{eq:dFI}, respectively. For $I_{z_+}$, it suffices to give the estimates of $I_{z_+,k}$, $k=1,\ldots,4$, which will be done in the following lemma.
\begin{lemma}\label{I-z+-expre}
With the integrals $I_{z_+,k}$, $k=1,\ldots,4$, defined in \eqref{eq:z+1}--\eqref{eq:z+4}, we have, as $t\to+\infty$,
\begin{align}
I_{z_+,1}&=-\frac{\rho}{h_1}\int_{[0,-\infty)}
\mathrm{Tr}\left[\Phi_{(\mathrm{PC})}^{-1}(\eta)\Phi_{(\mathrm{PC})}'(\eta)
\sigma_-\right]\frac{\mathrm{d}\eta}{2\pi i}+O\left(t^{-\frac{1}{2}}\right),\label{eq:estIz+1}\\
I_{z_+,2}&=\frac{h_1}{\rho}\int_{i[0,+\infty)}
\mathrm{Tr}\left[\Phi_{(\mathrm{PC})}^{-1}(\eta)\Phi_{(\mathrm{PC})}'(\eta)
\sigma_+\right]\frac{\mathrm{d}\eta}{2\pi i}+O\left(t^{-\frac{1}{2}}\right),\label{eq:estIz+2}\\
I_{z_+,3}&=\frac{h_1(1+\rho^2)}{\rho(1-\rho^2)^2}\int_{i[0,-\infty)}
\mathrm{Tr}\left[\Phi_{(\mathrm{PC})}^{-1}(\eta)\Phi_{(\mathrm{PC})}'(\eta)
\sigma_-\right]\frac{\mathrm{d}\eta}{2\pi i}+O\left(t^{-\frac{1}{2}}\right),\label{eq:estIz+3}\\
I_{z_+,4}&=-\frac{\rho(1+\rho^2)}{h_1(1-\rho^2)^2}\int_{[0,+\infty)}
\mathrm{Tr}\left[\Phi_{(\mathrm{PC})}^{-1}(\eta)\Phi_{(\mathrm{PC})}'(\eta)
\sigma_+\right]\frac{\mathrm{d}\eta}{2\pi i}+O\left(t^{-\frac{1}{2}}\right),\label{eq:estIz+4}
\end{align}
where $h_1=\sqrt{2\pi}\mathrm{e}^{i\pi\nu}/\Gamma(-\nu)$ with $\nu$ given by \eqref{nu}.
\end{lemma}
\begin{proof}
If $z\in\mathcal{D}(z_+)$, it follows from \eqref{Final-transformation} that $T(z)=R(z)P^{(+)}(z)$, which gives us
$$
\mathrm{Tr}\left[T^{-1}(z)T'(z)\sigma_-\right]=
\mathrm{Tr}\left[(P^{(+)})^{-1}(z)(P^{(+)})'(z)\sigma_-\right]
+\mathrm{Tr}\left[(P^{(+)})^{-1}(z)R^{-1}(z)R'(z)P^{(+)}(z)\sigma_-\right].
$$
Recall that (see \eqref{P+})
$$
P^{(+)}(z)=E(z)\Phi_{(\mathrm{PC})}(t^{\frac{1}{2}}\zeta(z))
\left(\frac{h_1}{\rho}\right)^{\frac{\sigma_3}{2}}\sigma_3\mathrm{e}^{tg(z)\sigma_3},
$$
thus,
\begin{align*}
&\mathrm{Tr}\left[(P^{(+)})^{-1}(P^{(+)})'\sigma_-\right]=-\frac{\rho}{h_1}\mathrm{e}^{-2tg(z)}
\left\{\mathrm{Tr}\left[\Phi_{(\mathrm{PC})}^{-1}E^{-1}E'\Phi_{(\mathrm{PC})}\sigma_-\right]
+t^{\frac{1}{2}}\zeta'\mathrm{Tr}\left[\Phi_{(\mathrm{PC})}^{-1}
\Phi_{(\mathrm{PC})}'\sigma_-\right]\right\},\\
&\mathrm{Tr}\left[(P^{(+)})^{-1}R^{-1}R'P^{(+)}\sigma_-\right]=
-\frac{\rho}{h_1}\mathrm{e}^{-2tg(z)}
\mathrm{Tr}\left[\Phi_{(\mathrm{PC})}^{-1}E^{-1}R^{-1}R'E\Phi_{(\mathrm{PC})}\sigma_-\right],
\end{align*}
where $\Phi_{(\mathrm{PC})}$ and $\Phi_{(\mathrm{PC})}'$ are evaluated at $t^{\frac{1}{2}}\zeta(z)$ and the other functions are evaluated at $z$. Substituting the above formulas into the definition of $I_{z_+,1}$ in \eqref{eq:z+1} yields
\begin{align*}
I_{z_+,1}&=-\frac{\rho}{h_1}\int_{\Sigma_1\cap\mathcal{D}(z_+)}
t^{\frac{1}{2}}\zeta'(z)\mathrm{Tr}\left[\Phi_{(\mathrm{PC})}^{-1}(t^{\frac{1}{2}}\zeta(z))
\Phi_{(\mathrm{PC})}'(t^{\frac{1}{2}}\zeta(z))\sigma_-\right]\frac{\mathrm{d}z}{2\pi i}\\
&~~~-\frac{\rho}{h_1}\int_{\Sigma_1\cap\mathcal{D}(z_+)}
\mathrm{Tr}\left[\Phi_{(\mathrm{PC})}^{-1}(t^{\frac{1}{2}}\zeta(z))E^{-1}(z)E'(z)
\Phi_{(\mathrm{PC})}(t^{\frac{1}{2}}\zeta(z))\sigma_-\right]\frac{\mathrm{d}z}{2\pi i}\\
&~~~-\frac{\rho}{h_1}\int_{\Sigma_1\cap\mathcal{D}(z_+)}
\mathrm{Tr}\left[\Phi_{(\mathrm{PC})}^{-1}(t^{\frac{1}{2}}\zeta(z))E^{-1}(z)R^{-1}(z)
R'(z)E(z)\Phi_{(\mathrm{PC})}(t^{\frac{1}{2}}\zeta(z))\sigma_-\right]\frac{\mathrm{d}z}{2\pi i}.
\end{align*}

Let $\eta=t^{\frac{1}{2}}\zeta(z)$. Since $\zeta$ is injective in $\mathcal D(z_+)$, we have $z=\zeta^{-1}(t^{-\frac{1}{2}}\eta)$ and $\dd \eta=t^{\frac{1}{2}}\zeta'(z)\dd z$. By \eqref{zeta'behavior}, it is also easily seen that $\Sigma_1 \cap \mathcal D(z_+)$ is mapped onto $[0,s]$ by $\zeta=\zeta(z)$, where $s:=\zeta(d)<0$ with $d:=\Sigma_1\cap\partial\mathcal{D}(z_+)$. After this change of variable, $I_{z_+,1}$ can be rewritten as
\begin{align}
I_{z_+,1}&=-\frac{\rho}{h_1}\int_{[0,t^{\frac{1}{2}}s]}
\mathrm{Tr}\left[\Phi_{(\mathrm{PC})}^{-1}(\eta)\Phi_{(\mathrm{PC})}'(\eta)
\sigma_-\right]\frac{\mathrm{d}\eta}{2\pi i}\nonumber \\
&~~~-\frac{\rho t^{-\frac{1}{2}}}{h_1}\int_{[0,t^{\frac{1}{2}}s]}\zeta'(z)^{-1}
\mathrm{Tr}\left[\Phi_{(\mathrm{PC})}^{-1}(\eta)E^{-1}(z)E'(z)\Phi_{(\mathrm{PC})}(\eta)
\sigma_-\right]\frac{\mathrm{d}\eta}{2\pi i}\nonumber \\
&~~~-\frac{\rho t^{-\frac{1}{2}}}{h_1}\int_{[0,t^{\frac{1}{2}}s]}\zeta'(z)^{-1}
\mathrm{Tr}\left[\Phi_{(\mathrm{PC})}^{-1}(\eta)E^{-1}(z)R^{-1}(z)R'(z)E(z)
\Phi_{(\mathrm{PC})}(\eta)\sigma_-\right]\frac{\mathrm{d}\eta}{2\pi i}. \label{eq:Iz+1est1}
\end{align}

To proceed, we observe from the large-$\eta$ behavior of $\Phi_{(\mathrm{PC})}$ in
\eqref{PCAsyatinfty} and the definition of $E$ in \eqref{E-expression} that
$$
\Phi_{(\mathrm{PC})}(\eta)\sigma_-\Phi_{(\mathrm{PC})}^{-1}(\eta)
=O\left(\mathrm{e}^{-\frac{\eta^2}{2}}\right), \qquad \eta\to\infty
$$
and
\begin{equation}\label{E-estimate}
E^{\pm1}(z)=O(1),\quad E'(z)=O(1),\qquad t\to+\infty,
\end{equation}
where we have also made use of the fact that $\nu\in{i\mathbb{R}}$. This, together with the estimates of $R$, $R'$ given in \eqref{R-estimation}, \eqref{dR-estimation}, and the cyclic property of the trace, implies that as $t\to+\infty$,
\begin{align*}
&-\frac{\rho t^{-\frac{1}{2}}}{h_1}\int_{[0,t^{\frac{1}{2}}s]}\zeta'(z)^{-1}
\mathrm{Tr}\left[\Phi_{(\mathrm{PC})}^{-1}(\eta)E^{-1}(z)E'(z)\Phi_{(\mathrm{PC})}(\eta)
\sigma_-\right]\frac{\mathrm{d}\eta}{2\pi i}=O\left(t^{-\frac{1}{2}}\right),\\
&-\frac{\rho t^{-\frac{1}{2}}}{h_1}\int_{[0,t^{\frac{1}{2}}s]}\zeta'(z)^{-1}
\mathrm{Tr}\left[\Phi_{(\mathrm{PC})}^{-1}(\eta)E^{-1}(z)R^{-1}(z)R'(z)E(z)
\Phi_{(\mathrm{PC})}(\eta)\sigma_-\right]\frac{\mathrm{d}\eta}{2\pi i}
=O\left(t^{-1}\right),
\end{align*}
and
\begin{multline*}
-\frac{\rho}{h_1}\int_{[0,t^{\frac{1}{2}}s]}
 \mathrm{Tr}\left[\Phi_{(\mathrm{PC})}^{-1}(\eta)\Phi_{(\mathrm{PC})}'(\eta)
\sigma_-\right]\frac{\mathrm{d}\eta}{2\pi i}\\
=\frac{\rho}{h_1}\int_{(-\infty,0]}
\mathrm{Tr}\left[\Phi_{(\mathrm{PC})}^{-1}(\eta)\Phi_{(\mathrm{PC})}'(\eta)
\sigma_-\right]\frac{\mathrm{d}\eta}{2\pi i}+O\left(\mathrm{e}^{-c_4t}\right),
\end{multline*}
where $c_4>0$ is some constant. Inserting the above three estimates into \eqref{eq:Iz+1est1} then gives us \eqref{eq:estIz+1}.

The estimates in \eqref{eq:estIz+2}--\eqref{eq:estIz+4} can be proved in a similar manner, and we omit the details here.
\end{proof}

In view of \eqref{I-z+}, Lemma \ref{I-z+-expre} and the fact that the parabolic cylinder parameterix $\Phi_{(\mathrm{PC})}$ only depends on the parameter $\rho$ through $\nu$ in \eqref{nu}, the following estimate of $I_{z+}$ is immediate.
\begin{lemma}\label{I-z+-asympt}
With the integral $I_{z+}$ defined in \eqref{I-z+}, we have, as $t\to+\infty$,
\begin{equation*}
I_{z_+}=\mathcal{C}_1(\rho)+O\left(t^{-\frac{1}{2}}\right),
\end{equation*}
where $\mathcal{C}_1$ is a constant dependent on $\rho$ but independent of $n$ and the parameters $\tau_1,\ldots,\tau_{n-1}$.
\end{lemma}

To estimate $I_2$ in \eqref{I-2}, we denote $z_*:=\Sigma_5\cap\mathcal{D}(z_+)$ and split $I_2$ into two parts:
\begin{equation}\label{eq:I2decomp}
I_2=-\frac{4\rho}{1-\rho^2}\int_{[0,z_+]}\mathrm{Tr}
\left[T^{-1}(z)T'(z)\sigma_3\right]\frac{\mathrm{d}z}{2\pi i}:=I_{2,1}+I_{2,2},
\end{equation}
where
\begin{align}
I_{2,1}&=-\frac{4\rho}{1-\rho^2}\int_{[0,z_*]}\mathrm{Tr}
\left[T^{-1}(z)T'(z)\sigma_3\right]\frac{\mathrm{d}z}{2\pi i}, \label{def:I21}\\
I_{2,2}&=-\frac{4\rho}{1-\rho^2}\int_{[z_*,z_+]}\mathrm{Tr}
\left[T^{-1}(z)T'(z)\sigma_3\right]\frac{\mathrm{d}z}{2\pi i}\label{def:I22}.
\end{align}
The estimate of $I_{2,1}$ is relatively easy, as given below.
\begin{lemma}\label{I-21-asympt}
With the integral $I_{2,1}$ defined in \eqref{def:I21}, we have, as $t\to+\infty$,
\begin{equation}\label{eq:estI21}
I_{2,1}=\frac{4\nu\rho}{(1-\rho^2)\pi i}\ln\left|\frac{z_*-z_+}{z_*-z_-}\right|
+O\left(t^{-\frac{1}{2}}\right).
\end{equation}
\end{lemma}
\begin{proof}
If $z\in[0,z_*]$, it follows from \eqref{Final-transformation} that $T(z)=R(z)P^{(\infty)}(z)$. Since $R$ satisfies the estimates \eqref{R-estimation} and \eqref{dR-estimation},  we have, as $t\to+\infty$,
\begin{align*}
\mathrm{Tr}\left[T^{-1}(z)T'(z)\sigma_3\right]&=
\mathrm{Tr}\left[(P^{(\infty)})^{-1}(z)(P^{(\infty)})'(z)\sigma_3\right]
+\mathrm{Tr}\left[(P^{(\infty)})^{-1}(z)R^{-1}(z)R'(z)P^{(\infty)}(z)\sigma_3\right]\\
&=\mathrm{Tr}\left[(P^{(\infty)})^{-1}(z)(P^{(\infty)})'(z)\sigma_3\right]+
O\left(t^{-\frac{1}{2}}\right), \qquad z\in[0,z_*].
\end{align*}
In view of the definition  of $P^{(\infty)}$ in \eqref{Pinfty}, it is direct to calculate
$$
\mathrm{Tr}\left[(P^{(\infty)})^{-1}(P^{(\infty)})'\sigma_3\right]
=-2\nu\left(\frac{1}{z-z_+}-\frac{1}{z-z_-}\right).
$$
A combination of the above two formulas and \eqref{def:I21} gives us
\begin{align*}
I_{2,1}&=\frac{8\nu\rho}{1-\rho^2}\int_{[0,z_*]}
\left(\frac{1}{z-z_+}-\frac{1}{z-z_-}\right)
\frac{\mathrm{d}z}{2\pi i}+O\left(t^{-\frac{1}{2}}\right), \qquad t\to +\infty,
\end{align*}
which is \eqref{eq:estI21} after integration.
\end{proof}

The estimate of $I_{2,2}$ relies on the following decomposition.
\begin{lemma}\label{I-22}
With the integral $I_{2,2}$ defined in \eqref{def:I22}, we have, as $t\to+\infty$,
\begin{equation}\label{eq:decompI22}
I_{2,2}=I^{(1)}_{2,2}+I^{(2)}_{2,2}+I^{(3)}_{2,2}+O\left(t^{-\frac{1}{2}}\right),
\end{equation}
where
\begin{align}
I^{(1)}_{2,2}&=-\frac{8t\rho }{1-\rho^2}\int_{[z_*,z_+]}g'(z)\frac{\mathrm{d}z}{2\pi i}=\frac{4\rho[g(z_*)-g(z_+)]}{(1-\rho^2)\pi i}t,\label{def:I221}\\
I^{(2)}_{2,2}&=-\frac{4t^{\frac{1}{2}}\rho }{1-\rho^2}\int_{[z_*,z_+]}
\zeta'(z)\mathrm{Tr}\left[\Phi_{(\mathrm{PC})}^{-1}(t^{\frac{1}{2}}\zeta(z))
\Phi_{(\mathrm{PC})}'(t^{\frac{1}{2}}\zeta(z))\sigma_3\right]\frac{\mathrm{d}z}{2\pi i},
\label{def:I222}\\
I^{(3)}_{2,2}&=-\frac{4\rho}{1-\rho^2}\int_{[z_*,z_+]}\mathrm{Tr}
\left[\Phi_{(\mathrm{PC})}^{-1}(t^{\frac{1}{2}}\zeta(z))E^{-1}(z)E'(z)
\Phi_{(\mathrm{PC})}(t^{\frac{1}{2}}\zeta(z))\sigma_3\right]\frac{\mathrm{d}z}{2\pi i}.\label{def:I223}
\end{align}
\end{lemma}
\begin{proof}
If $z\in\mathcal{D}(z_+)$, it follows from \eqref{Final-transformation} that $T(z)=R(z)P^{(+)}(z)$, which implies
$$
\mathrm{Tr}\left[T^{-1}(z)T'(z)\sigma_3\right]=
\mathrm{Tr}\left[(P^{(+)})^{-1}(z)(P^{(+)})'(z)\sigma_3\right]
+\mathrm{Tr}\left[(P^{(+)})^{-1}(z)R^{-1}(z)R'(z)P^{(+)}(z)\sigma_3\right].
$$
Using the definition of $P^{(+)}$ in \eqref{P+}, we further obtain
\begin{align*}
&\mathrm{Tr}\left[(P^{(+)})^{-1}(P^{(+)})'\sigma_3\right]=2tg'+
\mathrm{Tr}\left[\Phi_{(\mathrm{PC})}^{-1}E^{-1}E'\Phi_{(\mathrm{PC})}\sigma_3\right]
+t^{\frac{1}{2}}\zeta'\mathrm{Tr}\left[\Phi_{(\mathrm{PC})}^{-1}
\Phi_{(\mathrm{PC})}'\sigma_3\right],\\
&\mathrm{Tr}\left[(P^{(+)})^{-1}R^{-1}R'P^{(+)}\sigma_3\right]=
\mathrm{Tr}\left[\Phi_{(\mathrm{PC})}^{-1}E^{-1}R^{-1}R'E\Phi_{(\mathrm{PC})}\sigma_3\right],
\end{align*}
where again $\Phi_{(\mathrm{PC})}$ and $\Phi_{(\mathrm{PC})}'$ are evaluated at $t^{\frac{1}{2}}\zeta(z)$ and the other functions are evaluated at $z$. As a consequence, it is readily seen that
\begin{align}\label{eq:I22decomp1}
&\int_{[z_*,z_+]}\mathrm{Tr}\left[T^{-1}(z)T'(z)\sigma_3\right]\frac{\mathrm{d}z}{2\pi i}
\nonumber \\
&=2t\int_{[z_*,z_+]}\frac{g'\dd z}{2\pi i}+\int_{[z_*,z_+]}\mathrm{Tr}
\left[\Phi_{(\mathrm{PC})}^{-1}E^{-1}E'\Phi_{(\mathrm{PC})}\sigma_3\right]\frac{\mathrm{d}z}{2\pi i}+t^{\frac{1}{2}}\int_{[z_*,z_+]}\zeta'\mathrm{Tr}\left[\Phi_{(\mathrm{PC})}^{-1}
\Phi_{(\mathrm{PC})}'\sigma_3\right]\frac{\mathrm{d}z}{2\pi i}
\nonumber \\
&~~~
+\int_{[z_*,z_+]}\mathrm{Tr}\left[\Phi_{(\mathrm{PC})}^{-1}E^{-1}R^{-1}R'E
\Phi_{(\mathrm{PC})}\sigma_3\right]\frac{\mathrm{d}z}{2\pi i}.
\end{align}
With the aid of the large-$\eta$ behavior of $\Phi_{(\mathrm{PC})}$ given in
\eqref{PCAsyatinfty}, it is readily seen that
\begin{equation}\label{PC-PC-sigma3}
\Phi_{(\mathrm{PC})}(\eta)\sigma_3\Phi_{(\mathrm{PC})}^{-1}(\eta)=O(1),\qquad \eta \to \infty,
\end{equation}
uniformly for $\eta$ in the complex plane. This, together with the  estimates of $R$, $R'$, $E^{-1}$ and $E'$ given in \eqref{R-estimation},   \eqref{dR-estimation} and  \eqref{E-estimate}, implies that
\begin{multline*}
\int_{[z_*,z_+]}\mathrm{Tr}\left[\Phi_{(\mathrm{PC})}^{-1}E^{-1}R^{-1}R'E
\Phi_{(\mathrm{PC})}\sigma_3\right]\frac{\mathrm{d}z}{2\pi i}
\\
=\int_{[z_*,z_+]}\mathrm{Tr}\left[E^{-1}R^{-1}R'E
\Phi_{(\mathrm{PC})}\sigma_3\Phi_{(\mathrm{PC})}^{-1}\right]\frac{\mathrm{d}z}{2\pi i}=O\left(t^{-\frac{1}{2}}\right),\qquad t\to+\infty.
\end{multline*}
Inserting the above estimate into \eqref{eq:I22decomp1}, we finally obtain \eqref{eq:decompI22}--\eqref{def:I223} from the definition of $I_{2,2}$ given in \eqref{def:I22}.
\end{proof}
The next two lemmas deal with the estimates of $I^{(2)}_{2,2}$ and $I^{(3)}_{2,2}$ in \eqref{eq:decompI22}.
\begin{lemma}\label{I-22-2}
With the integral $ I^{(2)}_{2,2}$ defined in \eqref{def:I222}, we have, as $t\to+\infty$,
\begin{align}\label{eq:estI222}
I^{(2)}_{2,2}&=\frac{4\rho[g(z_+)-g(z_*)]}{(1-\rho^2)\pi i}t-\frac{2\nu\rho}{(1-\rho^2)\pi i}\ln{t}-\frac{4\nu\rho\ln|\zeta(z_*)|}{(1-\rho^2)\pi i}+\mathcal{C}_2(\rho)+O\left(t^{-\frac{1}{2}}\right),
\end{align}
where
\begin{equation}\label{eq:c2}
\mathcal{C}_2(\rho)=\frac{4\rho }{1-\rho^2}\int_{[0,+\infty)}\left(\mathrm{Tr}\left[\Phi_{(\mathrm{PC})}^{-1}(\eta)
\Phi_{(\mathrm{PC})}'(\eta)\sigma_3\right]-\left(\eta-\frac{2\nu}{\eta+e^{-\frac{\pi}{4}i}}\right)\right)
\frac{\mathrm{d}\eta}{2\pi i}
\end{equation}
is a constant dependent on $\rho$ but   independent of $n$ and the parameters $\tau_1,\ldots,\tau_{n-1}$.
\end{lemma}
\begin{proof}
Let $\eta=t^{\frac{1}{2}}\zeta(z)$ and denote $r=|\zeta(z_*)|$, we can rewrite $I^{(2)}_{2,2}$ as
\begin{align}
I^{(2)}_{2,2}&=\frac{4\rho}{1-\rho^2}\int_{\mathrm{e}^{-\frac{\pi i}{4}}[0,t^{\frac{1}{2}}r]}
\mathrm{Tr}\left[\Phi_{(\mathrm{PC})}^{-1}(\eta)
\Phi_{(\mathrm{PC})}'(\eta)\sigma_3\right]\frac{\mathrm{d}\eta}{2\pi i}
\nonumber \\
&=\frac{4\rho }{1-\rho^2}\Bigg(\int_{\mathrm{e}^{-\frac{\pi i}{4}}[0,t^{\frac{1}{2}}r]}
\left(\mathrm{Tr}\left[\Phi_{(\mathrm{PC})}^{-1}(\eta)
\Phi_{(\mathrm{PC})}'(\eta)\sigma_3\right]
-\left(\eta-\frac{2\nu}{\eta+\mathrm{e}^{-\frac{\pi i}{4}}}\right)\right)
\frac{\mathrm{d}\eta}{2\pi i} \nonumber \\
&\quad+\int_{\mathrm{e}^{-\frac{\pi i}{4}}[0,t^{\frac{1}{2}}r]}\left(\eta-\frac{2\nu}{\eta+\mathrm{e}^{-\frac{\pi i}{4}}}\right)\frac{\mathrm{d}\eta}{2\pi i} \Bigg). \label{eq:estI2221}
\end{align}
From the large-$\eta$ behavior of $\Phi_{(\mathrm{PC})}$ given in
\eqref{PCAsyatinfty}, it follows that
$$
\mathrm{Tr}\left[\Phi_{(\mathrm{PC})}^{-1}(\eta)\Phi_{(\mathrm{PC})}'(\eta)\sigma_3\right]
-\left(\eta-\frac{2\nu}{\eta+e^{-\frac{\pi}{4}i}}\right)=O\left(\eta^{-2}\right), \qquad \eta\to\infty.
$$
This, together with the analyticity of $\mathrm{Tr}\left[\Phi_{(\mathrm{PC})}^{-1}(\eta)\Phi_{(\mathrm{PC})}'(\eta)\sigma_3\right]$ for $\arg \eta\in[-\frac{\pi}{4}, 0]$ and Cauchy's theorem, implies that, as $t\to +\infty$,
\begin{align*}
&\frac{4\rho }{1-\rho^2}\int_{\mathrm{e}^{-\frac{\pi i}{4}}[0,t^{\frac{1}{2}}r]}\left(\mathrm{Tr}\left[\Phi_{(\mathrm{PC})}^{-1}(\eta)
\Phi_{(\mathrm{PC})}'(\eta)\sigma_3\right]-\left(\eta-\frac{2\nu}{\eta+e^{-\frac{\pi}{4}i}}\right)\right)
\frac{\mathrm{d}\eta}{2\pi i}\\
&=\frac{4\rho }{1-\rho^2}\int_{\mathrm{e}^{-\frac{\pi i}{4}}[0,+\infty)}\left(\mathrm{Tr}\left[\Phi_{(\mathrm{PC})}^{-1}(\eta)
\Phi_{(\mathrm{PC})}'(\eta)\sigma_3\right]-\left(\eta-\frac{2\nu}{\eta+e^{-\frac{\pi}{4}i}}\right)\right)
\frac{\mathrm{d}\eta}{2\pi i}+O\left(t^{-\frac{1}{2}}\right)
\\
&=\mathcal{C}_2(\rho) +O\left(t^{-\frac{1}{2}}\right),
\end{align*}
where $\mathcal{C}_2$ is defined in \eqref{eq:c2}. Also, straightforward calculations yield
\begin{align*}
&\int_{\mathrm{e}^{-\frac{\pi i}{4}}[0,t^{\frac{1}{2}}r]} \left(\eta-\frac{2\nu}{\eta+e^{-\frac{\pi}{4}i}}\right)
\frac{\mathrm{d}\eta}{2\pi i}
=\int_{[0,t^{\frac{1}{2}}r]}\left(-i\eta-\frac{2\nu}{\eta+1}\right)
\frac{\mathrm{d}\eta}{2\pi i}
=-\frac{r^2t}{4\pi }-\frac{\nu}{\pi i}\ln\left(t^{\frac{1}{2}}r+1\right)\\
&=-\frac{r^2t}{4\pi }-\frac{\nu}{2\pi i}\ln{t}-\frac{\nu\ln{r}}{\pi i}+O\left(t^{-\frac{1}{2}}\right),\qquad t\to+\infty.
\end{align*}
Inserting the above two formulas into \eqref{eq:estI2221}, we obtain the desired estimate \eqref{eq:estI222} by noting that $r=|\zeta(z_*)|$ and $r^2=4i[g(z_+)-g(z_*)]$.
\end{proof}

\begin{lemma}\label{I-22-3}
With the integral $I^{(3)}_{2,2}$ defined in \eqref{def:I223}, we have, as $t\to+\infty$,
\begin{equation}\label{eq:estI223}
I^{(3)}_{2,2}=-\frac{4\nu\rho}{(1-\rho^2)\pi i}\left(\ln\left|2z_+\zeta'(z_+)\right|
-\ln\left|\frac{\zeta(z_*)(z_*-z_-)}{z_*-z_+}\right|\right)+O\left(t^{-\frac{1}{4}}\right).
\end{equation}
\end{lemma}
\begin{proof}
By the cyclic property of trace, we have
\begin{align}\label{I-22-3-compute}
I^{(3)}_{2,2}&=-\frac{4\rho}{1-\rho^2}\int_{[z_*,z_+]}\mathrm{Tr}
\left[E^{-1}(z)E'(z)\Phi_{(\mathrm{PC})}(t^{\frac{1}{2}}\zeta(z))\sigma_3
\Phi_{(\mathrm{PC})}^{-1}(t^{\frac{1}{2}}\zeta(z))\right]\frac{\mathrm{d}z}{2\pi i}\nonumber\\
&=-\frac{4\rho}{1-\rho^2}\int_{[z_*,z_+]}\mathrm{Tr}
\left[E^{-1}(z)E'(z)\sigma_3\right]\frac{\mathrm{d}z}{2\pi i}\nonumber\\
&~~~-\frac{4\rho}{1-\rho^2}\int_{[z_*,z_+]}\mathrm{Tr}
\left[E^{-1}(z)E'(z)\sigma_3\left(\sigma_3\Phi_{(\mathrm{PC})}(t^{\frac{1}{2}}\zeta(z))
\sigma_3\Phi_{(\mathrm{PC})}^{-1}(t^{\frac{1}{2}}\zeta(z))-I\right)\right]
\frac{\mathrm{d}z}{2\pi i}.
\end{align}
Using the definition of $E$ given in \eqref{E-expression}, the first integral in \eqref{I-22-3-compute} can be evaluated explicitly as follows:
\begin{align}
\int_{[z_*,z_+]}\mathrm{Tr}
\left[E^{-1}(z)E'(z)\sigma_3\right]\frac{\mathrm{d}z}{2\pi i}
&=2\nu\int_{[z_*,z_+]}\left(\left(\ln|\zeta(z)|\right)'-\frac{1}{z-z_+}
+\frac{1}{z-z_-}\right)\frac{\mathrm{d}z}{2\pi i} \nonumber \\
&=\frac{\nu}{\pi i}\left(\ln\left|2z_+\zeta'(z_+)\right|
-\ln\left|\frac{\zeta(z_*)(z_*-z_-)}{z_*-z_+}\right|\right). \label{eq:estI2232}
\end{align}
To estimate the second integral in \eqref{I-22-3-compute}, we take $z_1\in(z_*,z_+)$ such that $z_+-z_1=O(t^{-\frac{1}{4}})$.
Thus,
\begin{align}
&\int_{[z_*,z_+]}\mathrm{Tr}
\left[E^{-1}(z)E'(z)\sigma_3\left(\sigma_3\Phi_{(\mathrm{PC})}(t^{\frac{1}{2}}\zeta(z))
\sigma_3\Phi_{(\mathrm{PC})}^{-1}(t^{\frac{1}{2}}\zeta(z))-I\right)\right]
\frac{\mathrm{d}z}{2\pi i}\nonumber \\
&=\int_{[z_*,z_1]}\mathrm{Tr}
\left[E^{-1}(z)E'(z)\sigma_3\left(\sigma_3\Phi_{(\mathrm{PC})}(t^{\frac{1}{2}}\zeta(z))
\sigma_3\Phi_{(\mathrm{PC})}^{-1}(t^{\frac{1}{2}}\zeta(z))-I\right)\right]
\frac{\mathrm{d}z}{2\pi i}\nonumber \\
&~~~+\int_{[z_1,z_+]}\mathrm{Tr}
\left[E^{-1}(z)E'(z)\sigma_3\left(\sigma_3\Phi_{(\mathrm{PC})}(t^{\frac{1}{2}}\zeta(z))
\sigma_3\Phi_{(\mathrm{PC})}^{-1}(t^{\frac{1}{2}}\zeta(z))-I\right)\right]
\frac{\mathrm{d}z}{2\pi i}\nonumber \\
&=O\left(t^{-\frac{1}{4}}\right)+O(z_+-z_1)=O\left(t^{-\frac{1}{4}}\right),
\qquad t\to+\infty, \label{eq:estI2233}
\end{align}
where we have made use of the estimate \eqref{E-estimate} and the asymptotics $$\sigma_3\Phi_{(\mathrm{PC})}(\eta)\sigma_3\Phi_{(\mathrm{PC})}^{-1}(\eta)
=I+O\left(\eta^{-1}\right),\qquad \eta\to\infty$$ in the second equality. The estimate \eqref{eq:estI223} then follows directly from \eqref{I-22-3-compute}--\eqref{eq:estI2233}.
\end{proof}
In the proof of Lemma \ref{I-22-3}, it is not essential that $z_+-z_1=O(t^{-\frac{1}{4}})$, but is enough for the  evaluation of the constant $\mathcal{C}$ in \eqref{Fn-asymp}.  The asymptotic expansion of $\ln{F}_n(x;\rho)$ with error term has been derived in \eqref{Fn-asymp}, the remaining task is to evaluate the constant term  therein. Therefore, we do not pursue the optimal choice of $z_1$ in the proof.

We are now ready to give the estimate of $I_2$. 
\begin{lemma}\label{I-2-asympt}
With the integral $I_2$ defined in \eqref{I-2}, we have, as $t\to+\infty$,
\begin{equation}\label{eq:estI2}
I_2=-\frac{2\nu\rho}{(1-\rho^2)\pi i}\ln(8nt)
+\mathcal{C}_2(\rho)+O\left(t^{-\min\left\{\frac{2}{2n+1},\frac{1}{4}\right\}}\right),
\end{equation}
where $\mathcal{C}_2$ is defined in \eqref{eq:c2}.
\end{lemma}
\begin{proof}
In view of \eqref{eq:I2decomp}--\eqref{def:I22}, we obtain from Lemmas \ref{I-21-asympt}--\ref{I-22-3} that, as $t\to+\infty$,
\begin{align}\label{I-2-asy}
I_2&=I_{2,1}+I_{2,2}=I_{2,1}+I^{(1)}_{2,2}+I^{(2)}_{2,2}+I^{(3)}_{2,2}
+O\left(t^{-\frac{1}{2}}\right)\nonumber\\
&=-\frac{2\nu\rho}{(1-\rho^2)\pi i}\ln{t}+\mathcal{C}_2(\rho)
-\frac{4\nu\rho}{(1-\rho^2)\pi i}\ln\left|2z_+\zeta'(z_+)\right|
+O\left(t^{-\frac{1}{4}}\right).
\end{align}
From the estimates of $z_+$ and $\zeta'(z_+)$ given in \eqref{z+-expansion} and \eqref{zeta'behavior}, it is readily seen that
\begin{equation*}
2z_+\zeta'(z_+)=2\sqrt{2n}\left[1+O\left(t^{-\frac{2}{2n+1}}\right)\right],
\quad\mathrm{as}~~t\to+\infty.
\end{equation*}
Substituting the above asymptotics into \eqref{I-2-asy} then gives us \eqref{eq:estI2}.
\end{proof}

Finally, we obtain the following asymptoitcs of $\ln{F}_n$.
\begin{lemma}
As $t\to+\infty$, we have
\begin{align}\label{F-asympt}
\ln{F}_n(x;\rho)=
-\frac{2\ln(1-\rho^2)}{\pi{i}}g(z_+)t+\frac{\ln^2(1-\rho^2)}{4\pi^2}\ln(8nt)
&+\ln\left[G(1+\nu)G(1-\nu)\right]\nonumber\\
&+O\left(t^{-\min\left\{\frac{2}{2n+1},\frac{1}{4}\right\}}\right),
\end{align}
where $G$ is the Barnes $G$-function.
\end{lemma}
\begin{proof}
A combination of Lemmas \ref{Fred-estimate}, \ref{I-z+-asympt} and \ref{I-2-asympt} gives us
\begin{equation}\label{Fred-a-1}
\partial_{\rho}\ln{F}_n(x;\rho)=
\frac{4\rho{g}(z_+)t}{(1-\rho^2)\pi{i}}-\frac{2\nu\rho\ln(8nt)}{(1-\rho^2)\pi i}
+\chi(\rho)+O\left(t^{-\min\left\{\frac{2}{2n+1},\frac{1}{4}\right\}}\right),
\qquad t\to+\infty,
\end{equation}
where
$$
\chi(\rho):=2\mathcal{C}_1(\rho)+\mathcal{C}_2(\rho)
$$
with $\mathcal{C}_1$ and $\mathcal{C}_2$ being the constants in Lemmas \ref{I-z+-asympt} and \ref{I-2-asympt}, respectively.
Clearly, the constant $\chi(\rho)$
is independent of $n$ and of the parameters $\tau_1,\ldots,\tau_{n-1}$.

Integrating both sides of \eqref{Fred-a-1} with respect to $\rho$ and using the fact that $\ln F_n(x;0)=0$, we have
\begin{equation}\label{F-asympt-new}
\ln{F}_n(x;\rho)=
-\frac{2\ln(1-\rho^2)}{\pi{i}}g(z_+)t+\frac{\ln^2(1-\rho^2)}{4\pi^2}\ln(8nt)
+\widetilde{\chi}(\rho)+O\left(t^{-\min\left\{\frac{2}{2n+1},\frac{1}{4}\right\}}\right),
\end{equation}
where $\widetilde{\chi}$ is the constant of integration. Note that $\widetilde{\chi}=\widetilde{\chi}(\rho)$ is also independent of the parameters $n,\tau_1,\ldots,\tau_{n-1}$, and the fact that
(cf. \cite{BB2018,BIP2019})
\begin{equation}\label{F-2}
\ln{F}_1(x;\rho)=
\frac{2\ln(1-\rho^2)}{3\pi}t+\frac{\ln^2(1-\rho^2)}{4\pi^2}\ln(8t)
+\ln[G(1+\nu)G(1-\nu)]+O\left(t^{-\frac{1}{2}}\right),
\end{equation}
it is then readily seen $\widetilde{\chi}=\ln[G(1+\nu)G(1-\nu)]$. This, together with \eqref{F-asympt-new}, gives us \eqref{F-asympt}.
%
 \end{proof}

Substituting the expansion of $g(z_+)$ given by \eqref{eq:gexpansion} into \eqref{F-asympt} and comparing the obtained asymptotics with \eqref{Fn-asymp}, it is readily seen that the undetermined constant $\mathcal {C}$ in \eqref{Fn-asymp} is given by
\begin{equation}\label{eq:C}
\mathcal{C}=\ln\left[G(1+\nu)G(1-\nu)\right]+\frac{\ln^2(1-\rho^2)}{4\pi^2}\ln(8n).
\end{equation}
Asymptotics of $F_n$ in \eqref{F-asymptotic-infty2} then follows directly from \eqref{Fn-asymp}, \eqref{eq:C} and the fact that $\beta=2i\nu$.
This completes the proof of Theorem \ref{thm:largegap}.


\subsection{Proof of Theorem \ref{thm1} }\label{proof1}
The proof is similar to the case of $0<\rho <1$, which is instead based on asymptotic analysis of the RH problem for $\Psi$ with $\rho>1$.  By unwrapping the transformations $\Psi\to{X}\to{Y}\to{T}\to\widetilde{R}$, we obtain from \eqref{q-expression}, \eqref{Fred-expression-to-x}, \eqref{Pinftyatinfty} and \eqref{Rexpan} that, as $x\to-\infty$,
\begin{equation}\label{q-expression1}
q_n\left((-1)^{n+1}x;\rho\right)=2i|x|^{\frac{1}{2n}}
\left((\widetilde{P}^{(\infty)}_1)_{12}+(\widetilde{R}_1)_{12}\right)
\end{equation}
and
\begin{equation}\label{Fred-expression1}
\frac{\mathrm{d}}{\mathrm{d}x}\ln F_n(x;\rho)=2i|x|^{\frac{1}{2n}}
\left((\widetilde{P}^{(\infty)}_1)_{11}+(\widetilde{R}_1)_{11}\right),
\end{equation}
where $\widetilde{P}^{(\infty)}_1$ and $\widetilde{R}_1$ are the residue terms in the large-$z$ expansions of $\widetilde{P}^{(\infty)}$ and $\widetilde{R}$, respectively.
On account of the estimate \eqref{Rtilde-estimation}, it is readily seen that
\begin{equation}\label{R111asymp}
(\widetilde{R}_1)_{11}=O\left(|x|^{-1-\frac{1}{2n}}\right),\qquad (\widetilde{R}_1)_{12}=O\left(|x|^{-1-\frac{1}{2n}}\right), \qquad x\to-\infty.
\end{equation}
%
Substituting $\widetilde{P}^{(\infty)}_1$ in \eqref{eq:tildeP1infty} and the above estimates into \eqref{q-expression1} and \eqref{Fred-expression1} gives
\begin{align}\label{q-expression2}
q_{n}((-1)^{n+1}x;\rho)=\frac{4iz_+}{c^{-1}-c}|x|^{\frac{1}{2n}}
+O\left(|x|^{-1}\right)
\end{align}
and
\begin{equation}\label{Fred-expression2}
\frac{\mathrm{d}}{\mathrm{d}x}\ln F_n(x;\rho)=2iz_+|x|^{\frac{1}{2n}}
\left(2\nu_0-\frac{c^{-1}+c}{c^{-1}-c}\right)
+O\left(|x|^{-1}\right),
\end{equation}
where $\nu_0$ and $c$ are given in \eqref{nu0} and \eqref{c}, respectively.  The error bounds here are uniformly valid for $x$ bounded away from the sequence of points $\{x_m\}_{m\in\mathbb{N}}$ determined by \eqref{poles-equation}. Since
\begin{equation}
c^{-1}-c=2i\sin(\omega(x,\rho)), \qquad  c^{-1}+c=2\cos(\omega(x,\rho)),
\end{equation}
where
\begin{align}\label{omega}
\omega(x,\rho):=2ig(z_+)|x|^{\frac{2n+1}{2n}}&-\frac{2n+1}{2n}i\nu_0\ln|x|\nonumber\\
&-i\nu_0\ln(4iz_+^2\zeta'(z_+)^2)+\arg\Gamma\left(\frac{1}{2}-\nu_0\right)
+\frac{\pi}{2},
\end{align}
we can rewrite \eqref{q-expression2} and \eqref{Fred-expression2} as
\begin{align}\label{q-expression3}
q_{n}((-1)^{n+1}x;\rho)=\frac{2z_+|x|^{\frac{1}{2n}}}{\sin(\omega(x,\rho))}
+O\left(|x|^{-1}\right)
\end{align}
and
\begin{equation}\label{Fred-expression3}
\frac{\mathrm{d}}{\mathrm{d}x}\ln F_n(x;\rho)=2z_+|x|^{\frac{1}{2n}}
[2i\nu_0-\cot(\omega(x,\rho))]
+O\left(|x|^{-1}\right).
\end{equation}
By substituting the expansions \eqref{z+-expansion}, \eqref{eq:gexpansion} and \eqref{zeta'behavior} into \eqref{q-expression3} and \eqref{Fred-expression3}, we finally arrive at the asymptotic approximations given in  \eqref{qAymp-2} and \eqref{F-asymptotic-infty}.

\subsection{Proof of Corollary \ref{cor}}\label{proof:cor}
To estimate locations of the poles of $q_{n}((-1)^{n+1}x;\rho)$, we need the following result about the zeros of a function (see \cite[Theorem 1]{Hethcote}).
\begin{lemma}\label{lem:zeros}
Let $J=[\lambda_0-\delta, \lambda_0+\delta]$. Assume that $f(\lambda)=h(\lambda)+\varepsilon(\lambda)$, where $f(\lambda)$ is continuous, $h(\lambda)$ is differentiable, $h(\lambda_0)=0$, $l=\min |h'(\lambda)| >0$, and
\begin{equation*}
M_\varepsilon:=\max|\varepsilon(\lambda)|<\min\Big\{|h(\lambda_0-\delta)|,
~|h(\lambda_0+\delta)|\Big\}.
\end{equation*}
Then, there exists a zero $\lambda_z$ of $f(\lambda)$  in the interval $J$ such that $|\lambda_z-\lambda_0|\leq M_\varepsilon/l$.
\end{lemma}

To make use of the above lemma, we note from \eqref{qAymp-2-r} that, if $\tau_1=\cdots=\tau_{n-1}=0$,
\begin{equation}\label{reciprocal-q}
L(x):=\frac{|x|^{\frac{1}{2n}}}{q_{n}((-1)^{n+1}x;\rho)} =\sin\Theta +O\left (|x|^{-\frac{2n+1}{2n}}\right ),~~\mathrm{as}~~x\to -\infty,
\end{equation}
where
\begin{equation}\label{def:Theta}
\Theta=\Theta(x):=\frac{2n}{2n+1}|x|^{\frac{2n+1}{2n}}
-\frac{2n+1}{2n}\kappa\ln|x|+\varphi,
\end{equation}
and where the constants $\kappa$ and $\varphi$ are given in \eqref{kappa}. The error bound in \eqref{reciprocal-q} is uniformly valid for $\Theta$ bounded away from $m\pi$, $m\in \mathbb{N}$. By \eqref{def:Theta}, it is readily seen that $x \sim -(\frac{2n+1}{2n})^{\frac{2n}{2n+1}} \Theta^{\frac{2n}{2n+1}}$ for large negative $x$. Thus, $\Theta$ can be taken as a large parameter. Applying Lemma \ref{lem:zeros} to the function $L$ in \eqref{reciprocal-q}, we see that for any integer $m$ large enough, there exists a zero $p_m$ of $L(x)$ such that
\begin{equation*}
\Theta(p_m)-m\pi=O(1/m),~~~~m\to\infty.
\end{equation*}
Equivalently, this means there exists a simple pole $x=p_m$ of $q_{n}((-1)^{n+1}x;\rho)$ satisfying the asymptotic approximation \eqref{poles}.

\begin{appendices}

\section{The parabolic cylinder  parametrix}\label{PCP}

The parabolic cylinder parametrix
$\Phi_\mathrm{(PC)}(\eta)=\Phi_\mathrm{(PC)}(\eta;\nu)$ with $\nu$  being a real or complex parameter is a solution of the
following RH problem.

\subsection*{RH problem for $\Phi_{(\mathrm{PC})}$}

\begin{description}
\item{(1)} $\Phi_\mathrm{(PC)}(\eta)$ is analytic for $\eta\in \mathbb{C}\setminus (\cup^5_{k=1}\Gamma_{k})$, where
$$
\Gamma_{k}=\left \{\eta\in\mathbb{C}\mid\arg \eta=\frac{k\pi}{2}\right\}, \quad k=1,\cdots,4,
   \qquad \Gamma_{5}=\left\{\eta\in\mathbb{C}\mid\arg \eta=-\frac{\pi}{4}\right\};
$$
   see Figure \ref{PC} for an illustration.

\begin{figure}[H]
  \centering
  \includegraphics[width=7cm,height=6cm]{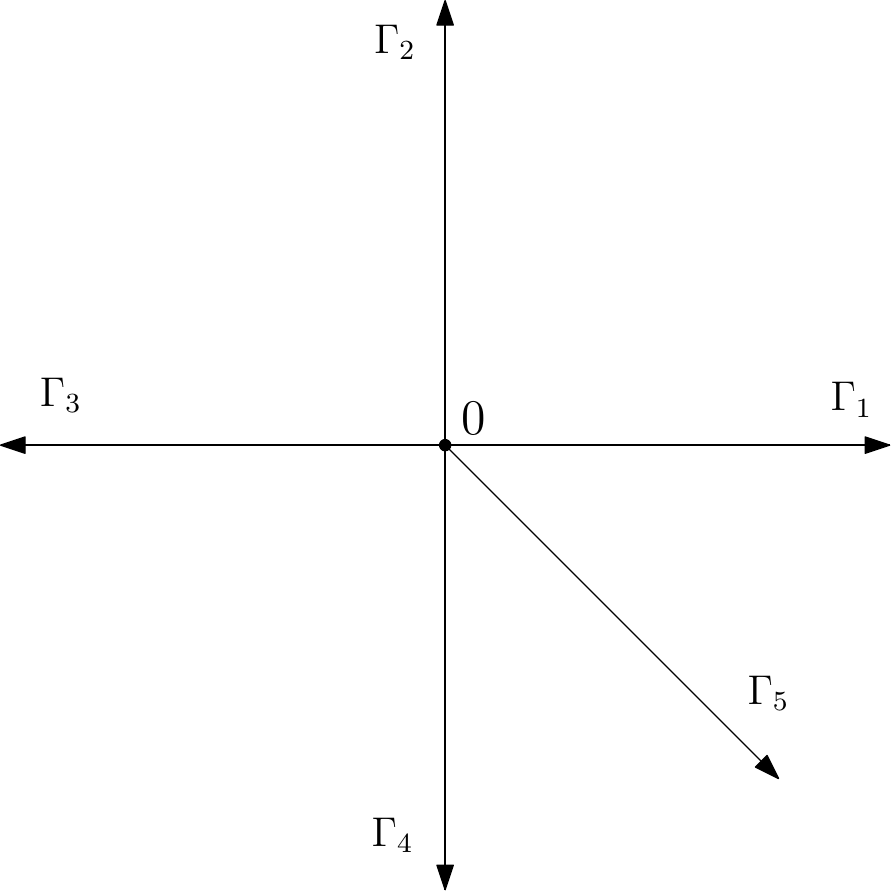}\\ 
  \caption{The jump contours $\Gamma_k$, $k=1,\ldots,5$, of the RH problem for $\Phi_{(\mathrm{PC})}$.}
\label{PC}
\end{figure}

\item{(2)} $\Phi_\mathrm{(PC)}(\eta)$ satisfies the jump condition
\begin{equation}\label{PCjumpmatrices}
\Phi_{(\mathrm{PC}),+}(\eta)=\Phi_{(\mathrm{PC}),-}(\eta)\left\{\begin{aligned}
&H_0,\quad && \eta \in\Gamma_1, \\
&H_1,\quad && \eta \in\Gamma_2, \\
&H_2,\quad && \eta \in\Gamma_3, \\
&H_3,\quad && \eta \in\Gamma_4, \\
&\mathrm{e}^{2\pi i\nu\sigma_3}, && \eta \in\Gamma_5,
\end{aligned}\right.
\end{equation}
where
$$
H_{0}=\begin{pmatrix}1 & 0 \\ h_{0} & 1\end{pmatrix}, \
H_{1}=\begin{pmatrix}1 & h_{1} \\ 0 & 1\end{pmatrix}, \
H_{n+2}=\mathrm{e}^{i \pi\left(\nu+\frac{1}{2}\right) \sigma_{3}} H_{n}
\mathrm{e}^{-i \pi\left(\nu+\frac{1}{2}\right) \sigma_{3}}, \ n=0,1
$$
with
\begin{equation}\label{h0}
h_{0}=-i \frac{\sqrt{2 \pi}}{\Gamma(\nu+1)}, \qquad h_{1}=\frac{\sqrt{2 \pi}}{\Gamma(-\nu)}\mathrm{e}^{i\pi\nu}, \qquad 1+h_{0}h_{1}=\mathrm{e}^{2\pi i\nu}.
\end{equation}
\item{(3)}
As $\eta \to \infty$, we have
\begin{align}\label{PCAsyatinfty}
\Phi_\mathrm{(PC)}(\eta)=\begin{pmatrix}
1+\frac{\nu(\nu+1)}{2\eta^2}+O\left(\frac{1}{\eta^4}\right) & \frac{\nu}{\eta}+O\left(\frac{1}{\eta^3}\right) \\[.2cm] \frac{1}{\eta}+\frac{(\nu+1)(\nu+2)}{2\eta^3}+O\left(\frac{1}{\eta^5}\right) & 1-\frac{\nu(\nu-1)}{2\eta^2}+O\left(\frac{1}{\eta^4}\right)\end{pmatrix}
\mathrm{e}^{\left(\frac{\eta^{2}}{4}-\nu\ln\eta\right) \sigma_{3}}.
\end{align}
\end{description}

According to \cite[Chapter 9]{FIKN}, the solution to above RH problem can be built explicitly in terms of the parabolic cylinder functions (cf. \cite[Chapter 12]{NIST}) $D_{\nu}$ and $D_{-\nu-1}$. More precisely, we have
\begin{equation}\label{PCsolution}
\Phi_\mathrm{(PC)}(\eta)=\begin{pmatrix}\frac{\eta}{2} & 1 \\ 1 & 0\end{pmatrix}
\begin{pmatrix}
D_{-\nu-1}(i\eta) & D_{\nu}(\eta) \\[.2cm]
D_{-\nu-1}'(i\eta) & D_{\nu}'(\eta)\end{pmatrix}
\begin{pmatrix}
\mathrm{e}^{i \frac{\pi}{2}(\nu+1)} & 0 \\
0 & 1
\end{pmatrix}, \qquad \arg{\eta}\in\left(-\frac{\pi}{4},0\right).
\end{equation}
For $\eta$ in other sectors of the complex plane, $\Phi_\mathrm{(PC)}$ is determined by \eqref{PCsolution} and the jump relation \eqref{PCjumpmatrices}.

\end{appendices}

\section*{Acknowledgements}
The authors are grateful to the reviewers for their constructive comments and suggestions. The work of Shuai-Xia Xu was supported in part by the National Natural Science Foundation of China under grant numbers 11571376 and 11971492, and by Guangdong Basic and Applied Basic Research Foundation (Grant No. 2022B1515020063). Lun Zhang was supported in part by the National Natural Science Foundation of China under grant numbers 12271105 and 11822104, and by ``Shuguang Program'' supported by Shanghai Education Development Foundation and Shanghai Municipal Education Commission. Yu-Qiu Zhao was supported in part by the National Natural Science Foundation of China under grant numbers 11571375 and 11971489.

\end{document}